\documentclass[12pt, reqno]{amsart}

\usepackage[thm]{macros}
\newgeometry{margin=1.25in}
\title{Certified Decisions}
\author{Isaiah Andrews}
\author{Jiafeng Chen}

\date{This version: October 21, 2026; first version February 25, 2025. Andrews: Department of Economics, Massachusetts Institute of Technology,
 iandrews@mit.edu; Chen: Department of Economics, Stanford
 University, jiafeng@stanford.edu. We thank Yeon-Koo Che, Roberto Corrao, Jos\'e Ignacio Cuesta, Guido Imbens, Ivana Komunjer, Lihua Lei, Adam McCloskey, Jos\'e Montiel Olea, Adam Rosen, Aaron Roth,
 Jonathan Roth, Suproteem Sarkar, Liyang Sun, Chris Walker, Davide Viviano, and seminar and workshop participants at NYU, Columbia, University of Pennsylvania, Stanford, and SEA 2025 for useful discussions.
 We thank Giacomo Fraccaroli for superb research assistance and valuable insights.
 \emph{AI Disclosure}:  AI tools (GPT, Claude) were used throughout the research process, including writing code and drafting proofs. The paper itself is entirely written by the authors, and we take full responsibility for its content.}

\newcommand{\ext}{\operatorname{ext}}

\newcommand{\M}{\mathcal{M}}
\newcommand{\MDMP}{\mathcal{M}_{\mathrm{DM}}^{(P,\alpha)}}

\begin{document}
\maketitle

\begin{abstract}
Hypothesis tests and confidence sets are ubiquitous in empirical
 research, yet their connection to subsequent decision-making is often unclear. We view statistical inference as pairing recommended decisions with guarantees of their quality---or \emph{certificates}---and develop a theory of  \emph{certified decisions}. Specifically, we attach P-certificates---upper bounds on loss that hold with
 probability at least $1-\alpha$---to recommended actions.
  We prove that it is without loss to limit attention to
 P-certificates generated by minimax decisions over confidence sets, or what \citet
 {manski2021econometrics} terms ``as-if decisions with a set estimate.''
 We further show that such certificates
 allow safe adoption decisions for ambiguity-averse downstream
 decision-makers. A parallel
 argument applies to E-certified decisions obtained from e-values in settings with
 unbounded loss. We illustrate this view of statistical inference through an application to economic mobility and one to AI model routing.
\end{abstract}

\newpage

\section{Introduction}

Inferential tools, such as hypothesis tests and confidence intervals, are ubiquitous in
empirical practice.  At the same time, many empirical papers aim, directly or indirectly,
to improve decision-making by potential consumers of research, such as policymakers,
firms, or households.  As highlighted by \citet{manski2021econometrics}, the connection
between inference and decision-making is unclear: outside of special cases,
optimal decisions in the tradition of \citet{wald1949statistical} rarely depend on the data
through the results of conventional inference procedures.\footnote{See, however,
\citet{tetenov2016testing, batesetal2024testing,vivianoetal2025testing,epanomeritakis2026statistical} for some
interesting exceptions. Another exception is Bayesian inference, where both decisions and inferential conclusions are generated by the posterior; see \citet{erol2026coverage} for recent work that aims to reconcile frequentist confidence sets with Bayesian decision-making.}  There is thus an apparent tension between the goal of informing
decision-making and conventional, inference-focused, methods for empirical research.

Despite this tension---or perhaps because of it---some procedures in the literature bridge inference and decision-making. In particular, when researchers
report confidence sets, one natural recipe for decisions is to proceed \emph{as if} the
confidence set enumerates all plausible parameter values, and choose actions that minimize
the worst-case loss over the confidence set. That is, for $\mathcal A$ an action space,
$L(a, \theta)$ a loss function, and $\hat\Theta(Y) \subset \Theta$ a data-dependent
confidence set with $P\{\theta\in\hat\Theta(Y)\}\ge1-\alpha$, one may choose actions \[
\delta(Y) \in \argmin_{a\in \mathcal A} \sup_{\theta \in \hat\Theta(Y)} L(a, \theta).
\numberthis \label{eq:asif_p} \] \citet{manski2021econometrics} terms decision rules of
this form ``as-if decisions with set estimates.''\footnote{Early instances of such decisions appear in \citet{henry2012set,goldfarb2003robust}.} Making decisions based on confidence sets is implicitly envisioned by Neyman, who writes \citep[][emphasis ours]{neyman1977frequentist}
\begin{quote} The problem of {confidence intervals} consists in determining two functions
of the observables [, $Y_1(\cdot) \le Y_2(\cdot)$,] to be used in the following manner:
Whenever the observable variables \( X \) assume some values [$x$], we shall calculate the
corresponding values of \( Y_1 \) and \( Y_2 \), say \( Y_1(x) < Y_2(x) \), and
\textbf{then assert (or act on the assumption) that} \[ Y_1(x) \leq \theta \leq Y_2(x). \]
\end{quote} More recently,
\citet{ben2021safe,chernozhukov2025policylearningconfidence} and \citet{yamin2026and}
select policies by minimizing the worst-case loss over a confidence set for a social
welfare function.

Several recent papers have also considered the other direction, first making a
statistical decision and then constructing confidence sets for the resulting loss or welfare.
\citet{kitagawa2018ewm} select a policy by maximizing an empirical welfare function, and
in their supplementary materials discuss how one may construct a confidence set for the
true welfare via simultaneous inference and projection \citep[see also][]
{ponomarev2024lower}. The
recent literature studying inference on winners
\citep{benjamini2019winners,andrews2024inference,zrnic2025flexible,zrnic2024locally,cheng2025chiseling}
generalizes this approach, considering selection of the ``best'' of several options based
on noisy estimates of loss or welfare, and then constructing bounds for the true loss
associated with the estimated ``best'' choice. That is, for $\delta(Y)$ a decision rule chosen by optimizing some empirical criterion, this literature (in the one-sided case) constructs a loss bound $R(Y)$ with
\begin{equation}\label{eq: P-certificate}
P\{L(\delta(Y), \theta)\le R(Y)\}\ge 1-\alpha.
\end{equation}

These exercises---choosing actions based on confidence sets and
constructing confidence bounds for the losses of recommended actions---can be unified by viewing statistical inference as a means of \emph{certifying} the performance of proposed decisions. Imagine a data analyst who
provides a decision-maker with a recommended action and an estimated loss bound
$(\delta(Y), R (Y)).$ If this pair satisfies \eqref{eq: P-certificate}, we call $R(Y)$ a
P-certificate for $L(\delta(Y),\theta),$ and $(\delta(Y), R(Y))$ a \emph{P-certified
decision}. The loss bounds $R(Y)$ derived in the literature studying inference on winners
are thus P-certificates by construction. For as-if decision-making, on the other hand, note that if we define
\begin{equation}\label{eq: as-if R} R(Y)=\sup_ {\theta \in \hat\Theta(Y)} L(\delta(Y),
\theta) \end{equation} then $\theta\in\hat\Theta(Y)$ implies $R(Y)\ge
L(\delta(Y),\theta),$ so \[ P\br{L(\delta(Y), \theta) \le R(Y)} \ge P(\theta \in
\hat\Theta(Y)) \ge 1-\alpha.\] Hence, both procedures discussed above generate
P-certificates. This paper develops---and advocates for---a view in which the purpose of statistical inference is to provide guarantees for decisions.

 We first study the optimal provision of P-certified decisions. We find that
the class of as-if decisions \`a la \eqref
 {eq:asif_p} and \eqref{eq: as-if R} is \emph{essentially complete} under any preference ordering that favors tighter certificates.\footnote{This result echoes---though it
 is developed independently of---a recent result by \citet{kiyani2025decision}, in the
 context of conformal prediction.} That is, any P-certified decision $ (\tilde
\delta, \tilde R)$ can be weakly improved by \emph{some as-if decision} $(\delta, R)$
corresponding to some confidence set, in the sense that $R(\cdot) \le \tilde R(\cdot)$
almost surely and the loss bound $R(\cdot)$ is thus weakly tighter than the loss
bound $\tilde R(\cdot)$.  This establishes that as-if decisions are \emph{universal}: all reasonable P-certified decisions
can be cast as as-if optimization \eqref{eq:asif_p}. To be sure, the result does not distinguish certain
confidence sets from others.  Nevertheless, it has some bite, and implies for instance that the
studentized projection approach discussed in e.g.  \cite{andrews2024inference} is
dominated by the approaches of \citet{chernozhukov2025policylearningconfidence}.

With more structure on the decision problem, we further show that some optimality results for inference translate to optimality for P-certificates.  In particular, if
the parameter is scalar and the loss function prefers high parameter values (e.g., implementing
treatments with higher effects yields lower loss), then as-if decisions based on uniformly most accurate
confidence lower bounds for $\theta$ yield optimal P-certificates.

Next, to connect P-certified decisions to standard decision-theoretic notions such as risk and
expected utility, we study an ambiguity-averse decision-maker with limited capacity for
data analysis who delegates that analysis to an analyst. The analyst provides the
decision-maker with a P-certified decision. The decision-maker then needs to decide whether to adopt the action recommended
by the analyst, or to default to a safe action with known
(but potentially large) loss. This decision-maker believes that the analyst's certificate $R
(Y)$ contains information about the loss of the recommended action, but has limited
knowledge of the data-generating process and wishes to protect themselves against a large
increase in risk uniformly over a potentially large class of data-generating processes.

P-certificates are useful to such a decision-maker.  Specifically, when losses lie in
$[0,1]$ and the default action costs $C \in [0,1)$, any rule
that adopts the recommended action $\delta(Y)$ only when $R(Y)\le C$ yields expected loss
no greater than $C + \alpha (1-C)$, meaning that the decision-maker faces a risk increase
of at most $\alpha (1-C).$  Moreover, we characterize optimal adoption decisions for
decision-makers who balance subjective and worst-case risk. We find that an optimal adoption
rule either never adopts or interpolates between the threshold rule $\one(R(Y) \le C)$ and
the Bayes rule that adopts whenever subjective expected loss is below $C$.

We then illustrate how this perspective---viewing inference as providing certificates for
decisions---helps interpret empirical findings and construct recommendations in two
applications. First, we revisit Creating Moves to Opportunity
\citep[CMTO,][]{bergman2024cmto}, an experiment where housing voucher holders are
recommended putatively high-economic-mobility Census tracts. Since these Census tracts are
selected based on noisy estimates of their quality, there is concern about whether the voucher
holders were indeed induced to move to higher-quality places---especially with additional
ambiguity over their previous locations and their location choices.

\citet{andrews2024inference} and \citet{mogstad2024inference} report inference results with
starkly different qualitative conclusions on this question:  \citet{andrews2024inference} report that targeted tracts have significantly higher average mobility than non-targeted tracts, while \citet{mogstad2024inference} report that the tract-level mobility ranking is highly uncertain. We show that their inference
results are relevant for different decision problems, with different ambiguity sets over
movers' relocation patterns, and introduce a family of decision problems that interpolates
between them. We then construct certified decisions for this spectrum of decision problems. This exercise quantifies how much robustness to relocation patterns can be accommodated while still certifying positive mobility gains.

Our second application considers routing queries to large language models (LLMs) using
noisy estimates of model quality. We view different inference-on-winner-type procedures
\citep{zrnic2025flexible,zrnic2024locally,andrews2024inference,chernozhukov2025policylearningconfidence}
as providing both recommended actions and resulting certificates. The tightness of the
resulting certificate can be thought of as a notion of decision-relevant power for a given
procedure, which we thus compare across procedures. We find cases in
which recommending a model with a lower estimated payoff but greater precision yields a
more favorable guarantee than certifying the model with the highest estimated payoff. This
illustrates how choosing actions in a certificate-aware way---as \citet{chernozhukov2025policylearningconfidence} advocate---can be beneficial.

Finally, we extend our analysis to certified decisions that correspond to an alternative
approach to frequentist inference. A rapidly growing literature proposes \emph
{e-values} as alternatives to p-values and confidence sets \citep[see][for a review]
{ramdas2024hypothesis}. We show that e-variables generate as-if
decisions---called E-posterior decisions by
\citet{grunwald2023posterior}---that provide an
\emph{E-certificate} of the form $\E_P [L (\delta(Y),\theta) / R(Y)] \le 1$.  Similarly, we show
that the class of such as-if decisions is essentially complete among E-certified
decisions: for any E-certified $(\tilde\delta, \tilde R)$, there exists a collection of
e-variables and a corresponding E-posterior decision $(\delta(\cdot), R(\cdot))$ such that $R(\cdot) \le
\tilde R(\cdot)$ almost surely.
Such decisions imply risk bounds for downstream decision-makers even in contexts with unbounded loss functions. We illustrate these certificates in the LLM routing application.

This paper is most similar to recent work by \citet{kiyani2025decision}, who study
decision-making in the context of conformal inference. In work preceding ours,
\citet{kiyani2025decision} show that minimizing the expected value of a $1-\alpha$ upper
bound for loss is equivalent to choosing a conformal prediction set with coverage
$1-\alpha$ to minimize expected minimax loss
(see \eqref{eq:kiyani}). The key to both arguments is the observation that a
probabilistic upper bound on loss can be inverted to form a confidence set.

\Cref{sec:ambiguity_averse_decision-makers_and_p_certified_decisions} defines P-certified
decisions, shows how confidence sets generate them, studies their optimal provision, and analyzes their value to downstream ambiguity-averse
 decision-makers.
 \Cref{sec: p-certificate applications} illustrates our P-certificates in application, while \Cref{sec:extension_e_certificates} discusses the extension to E-certificates.

\section{P-certified decisions}
\label{sec:ambiguity_averse_decision-makers_and_p_certified_decisions}

Consider a statistical decision problem characterized by a loss function $L(a, \theta):
\mathcal A \times \Theta \to \R$ for an action $a \in \mathcal A$ and an unknown
parameter $\theta \in \Theta$.
Data $Y \in
\mathcal Y$ is drawn from some distribution $P$. To model uncertainty, following
\citet{andrews2025purpose}, we assume that Nature may choose any $ (\theta, P)$ pairs from
a set $\mathcal M$. For instance, if $(\theta, P)$ satisfy a statistical model indexed
by $\theta$, then $\mathcal M = \br{(\theta, P_\theta):\theta
\in \Theta}$.  More generally, we allow that $\theta$ may fail to fully describe $P,$ for
instance due to nuisance parameters or model incompleteness. This formulation also does
not impose point-identification: we may have $\br{(\theta_1, P), (\theta_2, P)} \subset
\M$ for $\theta_1 \neq \theta_2.$ Throughout, we refer to $L$ as a loss function, but it is possible
to choose $L(a, \theta) = \tilde L(a,\theta) - \inf_{a'\in \mathcal A} \tilde L(a',
\theta)$ so that $L$ represents regret for some primitive loss function $\tilde L$
instead.

We start by defining \emph{P-certified decisions} and establishing some basic properties.
\begin{defn}[P-certified decisions]
    For some $\alpha \in (0,1)$, the pair $(\delta(\cdot), R(\cdot)) : \mathcal Y \to
    \mathcal A
    \times \bar\R$ is called a $1-\alpha$ \emph{P-certified decision} if, for
    every $
    (\theta, P) \in \M$, \[
        P\pr{L(\delta(Y), \theta) \le R(Y)} \ge 1-\alpha.
        \numberthis
        \label{eq:p-cert}
    \]
    When this holds, $R(Y)$ is called a $(1-\alpha)$ \emph{P-certificate} for $\delta
    (Y)$.
\end{defn}

The guarantee \eqref{eq:p-cert} is familiar---demanding that $R(Y)$ be a
$(1-\alpha)$-confidence upper bound for the random quantity $L(\delta (Y), \theta)$. This
guarantee is a frequentist one: it holds ex ante over repeated draws of the data $Y$ for
fixed $(\theta, P)$, uniformly over the relevant econometric model $\M$. A realization of
a low $R(Y)$ certifies that the corresponding action $\delta(Y)$ is attractive,
in the sense that either its loss is low ($L(\delta(Y), \theta) \le R(Y)$), or something
improbable---specifically, an event with probability at most $\alpha$---has occurred.

We next study procedures that (optimally) generate P-certified decisions. We
then present a stylized delegation model in which P-certified decisions, generated by an
analyst recommending a decision, allow subsequent risk bounds for an ambiguity-averse
decision-maker who decides whether or not to accept analyst recommendations.

First, P-certified decisions are readily generated by as-if optimization with set
estimates \citep{goldfarb2003robust,henry2012set,manski2021econometrics,lim2026dynamically}.

\begin{restatable}{prop}{propcs}
\label{prop:cs}
Let $\hat\Theta(Y)$ be a confidence set: $P(\hat\Theta(Y) \ni \theta) \ge 1-\alpha$ for all $
(\theta, P) \in \M$. For any $\delta(Y)$, let $R(Y) = \sup_{\theta \in \hat\Theta(Y)} L
(\delta(Y), \theta)$. Then $(\delta, R)$ is a $(1-\alpha)$ P-certified decision. In
particular, if the minimum $\delta^*(Y) := \argmin_a \sup_{\theta \in \hat\Theta(Y)} L(a,
\theta)$
exists, then the corresponding $R^*(Y) := \sup_{\theta \in \hat\Theta(Y)} L
(\delta^*(Y), \theta) \le \sup_{\theta \in \hat\Theta(Y)} L
(\delta(Y), \theta)$ for any other $\delta$.\footnote{Throughout, the supremum over an empty set is equal to $-\infty$.}
\end{restatable}
\begin{proof}
Fix $(\theta, P) \in \M$.
    If $\theta \in \hat\Theta(Y)$ then $L(\delta(Y), \theta) \le R(Y)$ by definition. Thus
    $1-\alpha \le P(\theta \in \hat\Theta(Y)) \le P\br{L(\delta(Y), \theta) \le R(Y)}$ and
    $(\delta ,R)$ is a P-certified decision.
    For the second part, observe that $R^*(Y) = \min_a \sup_{\theta \in \hat\Theta(Y)} L
    (a, \theta) \le \sup_ {\theta \in \hat\Theta(Y)} L
(\delta(Y), \theta)$.\qedhere
\end{proof}

\Cref{prop:cs} shows that the worst-case loss of any decision rule over a given confidence
 set is automatically a P-certificate of that decision rule. Correspondingly, by
\emph{optimizing} this worst-case loss over a confidence set---exactly as in as-if
 optimization---we obtain the tightest certificate among those constructed
 this way. In
 \cref{sec:optimal_p_certificates}, we shall see that the
 converse is also
 essentially true: Any reasonable P-certified decision $(\delta, R)$  is (weakly) beaten by
 as-if optimization using \emph{some} confidence set.

We illustrate these results with a running example.

\begin{exsq}[Inference on winners]
\label{ex:iow}

We consider an example based on the recent literature on inference on winners
\citep{benjamini2019winners,andrews2024inference,zrnic2025flexible,zrnic2024locally}.
Suppose we are interested in potentially implementing an action $a$ from a finite set $
\mathcal{A}.$  For instance, we could be interested in whether to implement a program, as
studied in \citet{bergman2024cmto}, that encourages low-income households with children to
move to high-opportunity neighborhoods.  In this context, each $a\in\mathcal{A}$ could
index a potential set of ``recommended'' neighborhoods.

Our goal may be to maximize some outcome, for instance the average  household income rank
in adulthood (relative to the adult income distribution) for children in the targeted
households. Specifically, for $\theta(a)\in[0,1]$ the (unknown) true outcome associated
with selection $a,$ we would  like to choose actions $a$ that maximize $\theta(a)$. To represent this preference, we can consider the loss function $L
(a,\theta)=1-\theta(a)$.

Data (e.g. Census-tract level measures for economic opportunity   together with a model
for long-term outcomes as in \citealt{bergman2024cmto}) can be used to form an estimate $X(a)$ for the
outcome associated with each possible action $a,$ along with an associated standard error
$\sigma(a)$, which for simplicity we treat as known. For $a\in\mathcal{A}$, suppose that
$\frac{X(a)-\theta(a)} {\sigma(a)} \sim Z(a)$, for some random variables $\br{Z(a)}_{a \in \mathcal A}$ with a known distribution (e.g.,
multivariate Gaussian with known variance). In this setting, we can let $Y= (X (\cdot),\sigma (\cdot))$ collect the
estimates and standard errors.

One reasonable decision rule is empirical welfare maximization, which takes $\delta
(Y)=\argmax_{a\in\mathcal{A}}{X(a)}$.  This decision rule corresponds to as-if optimization against a
joint confidence set for $\theta$ based on an unstudentized statistic $X(a) - \theta(a)
= Z(a) \sigma(a)$: \[
    \hat\Theta
(Y) = \bigtimes_{a\in \mathcal A} [X (a) - c_ {1-\alpha}, \infty), \text{ where } P
\pr{\sup_ {a\in
\mathcal{A}}Z(a)\sigma(a) \le c_{1-\alpha}} = 1-\alpha \numberthis
\label{eq:unstudentized}
\]
This confidence set
 corresponds to an unstudentized version of what \citet{andrews2024inference} call the
 projection approach, which was suggested in the supplementary materials of \citet
 {kitagawa2018ewm}.  Its corresponding certificate is \[R(Y)=1 - X
 (\delta(Y)) + c_{1-\alpha} ,\]
 which inflates the estimated loss by $c_{1-\alpha}$.
  \end{exsq}

\subsection{Optimal provision of P-certificates}

\label{sec:optimal_p_certificates}

In settings where P-certificates are useful, it seems reasonable to prefer tighter bounds
on the loss.  That is, given the choice between level-$(1-\alpha)$ P-certificates $R(Y)$ and
$\tilde{R}(Y)$, if $R(Y) \le \tilde R (Y)$ for all $Y,$ it is natural to prefer $R$
to $\tilde{R}$. More generally, if we consider any loss function $B(R,\theta)$ that is
increasing in $R$ for all $\theta,$ then
$\E_P[B(R(Y),\theta)]\le\E_P[B(\tilde{R}(Y),\theta)]$ for all $(P,\theta),$ so $R$ yields
a lower expected loss than $\tilde{R}$.

When $R$ and $\tilde{R}$ are ordered almost surely, or more generally in the sense of
first-order stochastic dominance, we say $R$ weakly dominates $\tilde R$.
\begin{defn}
\label{defn:statewise}
We say that a $1-\alpha$ P-certified decision $(\delta, R)$ \emph
 {weakly dominates}  another $1-\alpha$ P-certified decision $
 (\tilde \delta, \tilde R)$  if \[
P( R(Y)\ge r) \le P( \tilde R(Y)\ge r)\]
for all $r$ and all $(\theta, P) \in\M$.
\end{defn}
We next show that the class of P-certified decisions from as-if optimization against a
confidence set, \`a la \cref{prop:cs}, is an \emph{essentially complete} class with respect to
dominance. That is, any P-certified decision rule $
(\tilde \delta, \tilde R)$ is weakly dominated by some decision rule $(\delta,R)$ formed by as-if optimization. As an implication, these decision
rules are also essentially complete with respect to any risk function of the form $\E_
{P} [B(R, \theta)]$, for any weakly increasing $B(\cdot, \theta)$.

To state this result, denote by $\mathcal C_{P,\alpha, \epsilon}$ the class of decision
rules formed by as-if optimization against a confidence set $\hat \Theta (Y)$ in the
following fashion: For $\epsilon > 0$,  consider $ (\delta (Y), R (Y))$ such that (i)
$\hat \Theta$ is a $1-\alpha$ confidence set under $\mathcal M$ and (ii)
\begin{align*}
\sup_{\theta \in \hat\Theta(Y)} L(\delta(Y), \theta) &\le \inf_{a \in \mathcal A} \sup_
{\theta \in \hat\Theta (Y)} L (a,
\theta) +
\epsilon  \quad\text{ and  }\quad R(Y) = \sup_{\theta \in \hat\Theta(Y)} L(\delta(Y), \theta)
\numberthis \label{eq:solving_asif}
\end{align*}
The constant $\epsilon$
 accounts for cases where the minimax loss over $\hat\Theta(Y)$ is not achieved by any element of $\mathcal A$; our result holds for any $\epsilon
 > 0$.
\begin{restatable}{theorem}{thmdominate}
\label{thm:dominate} Assume that
$\inf_{a\in\mathcal A}L(a,\theta)>-\infty$ for every $\theta\in\Theta$.
For any P-certified decision $(\tilde\delta, \tilde R)$ and any
 $\epsilon > 0$, there exists some $(\delta, R) \in \mathcal C_
 {P, \alpha, \epsilon}$ that weakly dominates $(\tilde \delta, \tilde R)$.
\end{restatable}

\Cref{thm:dominate} shows a sense in which it is without loss to limit attention to
 P-certificates derived from as-if optimization with confidence sets. Thus, not only is
 as-if optimization a convenient way to generate P-certified decisions, but there is, in a
 sense, no need to look beyond this recipe. \Cref{thm:dominate} follows from the simple observation that inverting the P-certification property yields a valid confidence set,
\[\tilde \Theta(Y) := \br{\theta : L(\tilde\delta(Y), \theta) \le \tilde R(Y)}, \quad P(\theta \in \tilde\Theta(Y)) = P(L(\tilde\delta(Y), \theta) \le \tilde R(Y)) \ge 1-\alpha.\] As-if optimizing against this confidence set can only weakly improve the resulting certificate.

\begin{rmksq}[Properties of as-if optimization] Since restricting to as-if optimization
 over confidence sets is, for the purpose of generating P-certified decisions, without
 loss, we highlight a few additional features of these decision rules. First,
 confidence sets allow for decision-makers with different decision problems---including ones unknown to the analyst---to derive
 their own P-certificates by optimizing \eqref{eq:solving_asif}, making confidence sets
 attractive for communicating information to heterogeneous decision-makers. Second, in
 addition to generating P-certificates, as-if decisions with confidence sets are
 dynamically consistent in the sense axiomatized by
 \citet{lim2026dynamically}.

Lastly, we comment on set-identification of $\theta$. The mechanics of as-if optimization,
its guarantees of P-certification, and its essential completeness do not qualitatively
depend on  whether $\theta$ is point- or set-identified.\footnote{Of course, construction
of confidence sets is often qualitatively different in point- versus set-identified
regimes. } This feature is distinct from some other problems considered in the partial
identification literature:
 \citet{christensen2026optimal}, for instance, argue that set-identification of payoffs meaningfully
  changes optimal decision rules in their setting. By contrast, so long as we consider decision problems for
  which $\theta$---and not the identified set---is the decision-relevant
  parameter,\footnote{One could instead define the loss function as relative to the
  identified set, with $L (a, \Theta_I(P)) := \sup_ {\theta
 \in \Theta_I(P)} \tilde L(a, \theta)$ as in \citet
  {christensen2026optimal} and \citet{henry2012set}. P-certification for such loss functions then requires
  confidence sets covering the identified set. } the coverage property that matters for
  P-certification is coverage for $\theta$ as in \citet{imbens2004confidence}, as opposed
  to coverage for the identified set \citep{chernozhukov2007estimation}.
\end{rmksq}

\begin{rmksq}[Relation to \citet{kiyani2025decision}]
    In the context of conformal prediction, recent work by \citet{kiyani2025decision}
 (their Theorem 2.3)  shows a result similar to \cref{thm:dominate}.\footnote
 {We developed \cref{thm:dominate} independently of their result. } To connect these results, let
 $\vartheta$ be a random prediction target and let $H$ be a joint distribution over $
 (\vartheta, X)$. In a prediction context, $\vartheta$ is the unknown true label of a
 unit and $X$ is the observed features. A conformal prediction set $\hat \Theta
 (X)$ satisfies $H(\vartheta \in \hat\Theta(X)) \ge 1-\alpha$, where the probability is
 taken over $(\vartheta, X) \sim H$. \citet{kiyani2025decision} show that the following
 optimization problems are equivalent, with essentially the same argument as that
 underlying \cref{thm:dominate}:
 \begin{align*}
  & \min_ {R (\cdot), \delta(\cdot)} \E_H [R (X)] \text{ subject to } H\br{L(\delta(X),
 \vartheta) \le R (X)} \ge 1-\alpha \\
& \min_{\hat\Theta(\cdot)} \E_H\bk{
    \min_{a\in \mathcal A} \max_{\vartheta \in \hat\Theta(X)} L(a, \vartheta)  } \text
     { subject to } H\{\vartheta \in \hat\Theta(X)\} \ge 1-\alpha.
     \numberthis \label{eq:kiyani}
 \end{align*} In contrast to \citet{kiyani2025decision}, we consider a parameter inference setting, where probabilities
  are solely taken over $Y$ fixing $\theta$, rather than jointly over $
  (X, \vartheta)$. We also consider evaluations of the certificate $R(\cdot)$ beyond its
  expectation, using a weaker preference ordering (\cref{defn:statewise}).
\end{rmksq}

\begin{exsq}[Inference on winners, continued]
\label{ex:iow_ctd}
\Cref{thm:dominate} establishes a strong sense in which as-if optimization is reasonable.
The dominance ordering in \cref{defn:statewise} is demanding, however, and there are many $(R,\delta)$,
$(\tilde{R},\tilde{\delta})$ pairs that cannot be ordered, in the sense that neither
$R(Y)$ nor $\tilde{R}(Y)$ (weakly) stochastically dominates the other.  Consequently, weak
dominance in the sense of \cref{defn:statewise} provides only a partial order over the
class of certified decisions, and does not generally provide clear recommendations for
``optimal'' procedures. This partial order is nevertheless sufficient to rule out some
methods discussed in the literature.

To illustrate, let us return to \cref{ex:iow}. There, we select a recommended action via empirical welfare maximization,
$\delta(Y)=\argmax_{a\in\mathcal{A}}X(a).$ We may P-certify this decision rule using
unstudentized projection as in \eqref{eq:unstudentized}, or we may use studentized
projection instead: \[
    \hat \Theta^*(Y) =\bigtimes_{a \in \mathcal A} [X(a) - c^*_{1-\alpha} \sigma(a),
    \infty),
    \text{ where } P\pr{\max_a Z(a)  \le c^*_{1-\alpha}} =
    1-\alpha,
\]
forming a P-certificate $R^*(Y) = \sup_{\theta \in \hat\Theta^*} L(\delta(Y), \theta) =
1-X
(\delta(Y)) + c_{1-\alpha}^* \sigma(\delta(Y))$ (with $R^*$ clipped to $[0,1]$). Studentization ensures that
the width of the interval reflects the standard error for the recommended policy $\delta
(Y),$ and further ensures that e.g. the inclusion of a single, very noisily estimated
option $a$ in the choice set does not greatly influence the critical value.  For this
reason, the theoretical discussion of projection in \citet{andrews2024inference} focuses
primarily on the studentized approach.

As discussed in \cref{ex:iow}, $R(Y)$ can
be understood via as-if optimization.
By contrast, $(\delta,R^*)$ is less naturally understood via as-if optimization. If the
analyst conducts as-if optimization over $\hat\Theta^* (Y)$, their  recommended decision is
$\tilde \delta(Y)=\argmax_{a\in\mathcal{A}}\br{X(a)-\sigma (a)c^*_{1-\alpha}},$ which is a
special case of the risk-aware empirical maximization approach proposed by
\citet{chernozhukov2025policylearningconfidence} and can differ from $\delta(Y)$ when
$\sigma(\cdot)$ is non-constant.  The corresponding risk bound
$\tilde{R}(Y)=1-\max_{a\in\mathcal{A}}\br{X(a)-\sigma(a)c^*_{1-\alpha}}$ is no larger than
$R^*(Y)$ in every realization of $Y$, so $(\tilde\delta,\tilde R)$ weakly dominates
$(\delta,R^*)$.\footnote{Formally, however, if $\Theta=[0,1]^\mathcal{A}$ then $
(\delta,R^*)$ still belongs to the essentially complete class we study, since it results
from as-if optimization against the confidence set $\hat\Theta^{**}(Y)=\br
{\theta\in\Theta:L(\delta(Y),\theta)\le 1-(X(\delta(Y))-\sigma(\delta(Y))c^*_
{1-\alpha})}.$  Note, however, that $\hat\Theta^{**}(Y)$ is a superset of $\hat\Theta^*
(Y),$ which helps explain the superior performance of $(\tilde\delta,\tilde {R})$ relative
to $(\delta,R^*)$.}
\end{exsq}

\subsubsection{Optimality under ordered $\theta$}

The non-uniqueness of optimal certified decisions, in the sense of weak dominance, is
intuitive: \Cref{thm:dominate} shows that P-certificates are closely linked to confidence
sets, and in most problems the class of admissible (i.e. undominated) confidence sets is
large.  If we narrow attention to settings where optimal (uniformly most accurate)
confidence sets exist, one might hope that optimal P-certified decisions will exist as
well and be found by as-if optimization.  The next theorem shows that hope is borne out,
provided the loss is monotone in the parameter and we limit our comparison of stochastic
dominance to risk values exceeding $ \inf_{a\in\mathcal{A}}L(a,\theta)$.\footnote{Absent
the latter restriction, it is impossible to dominate as-if optimization against the
trivial confidence set which takes $\hat\Theta(Y)$ empty with probability $\alpha,$ and
equal to $\Theta$ with probability $1-\alpha$.}

To state the result, we recall the definition of uniformly most accurate confidence bounds:
\begin{defn}[\citet{lehmann1986testing}, pp.~79--80]
    Suppose $\Theta \subseteq \R$. A random variable $\hat\theta(Y)$ is a $1-\alpha$
    \emph{confidence lower bound} if $P(\hat\theta(Y) \le \theta) \ge 1-\alpha$ for all $
     (\theta, P) \in \M$. Furthermore, it is \emph{uniformly most accurate} if for any
     $1-\alpha$ confidence lower bound $\tilde \theta(Y)$, any $t > 0$, and any $
     (\theta, P) \in \M$,
    \[P\{\hat\theta(Y)
    \le \theta - t\}
     \le P\{\tilde \theta(Y) \le \theta - t
    \}.\]
\end{defn}

\begin{restatable}{theorem}{thminc}
\label{thm:inc} Suppose $\Theta \subseteq \R$ is compact, $L(a, \theta)$ is weakly decreasing and right-continuous
in $\theta$ for each $a$, and $R(\theta)=\inf_{a\in \mathcal A} L(a, \theta)$ is achieved by some $a\in \mathcal
A$ for every $\theta \in \Theta$. Let $\hat\theta (Y)$ be a $1-\alpha$ \emph{uniformly
most accurate}  confidence lower bound, assumed to exist. Let $(\delta, R)$ be a certified
decision that as-if optimizes against $ [\hat\theta(Y), \infty) \cap \Theta$:
 \[R(Y) \equiv L(\delta(Y), \hat\theta(Y)) = \inf_{a\in \mathcal A} L(a, \hat\theta(Y)).\]
  Then for any other $1-\alpha$ P-certified $(\tilde\delta,
 \tilde R)$ and any $r>R(\theta),$
 \[P(R(Y)\ge r)\le P(\tilde{R}(Y)\ge r).\]
\end{restatable}

\begin{exsq}[Treatment choice]
\cref{thm:inc} does not apply to the inference on winners problem (\cref{ex:iow}) since
 the parameter of interest $\theta$ is vector-valued, and there does not in general exist
 a uniformly most accurate confidence bound.  Consider instead a setting where the
 decision-maker considers a single treatment, for instance enrollment in a job training
 program, and must decide whether to treat an entire population or not, $a \in \br{0,1}$.
 Suppose the loss function is $L(a, \theta) = -a \theta$ and the analyst observes a
 normally distributed estimate $Y \sim \Norm(\theta, \sigma^2)$ for known $\sigma^2$. Then
 the optimal certified decision in the sense of \cref{thm:inc} is $\delta(Y) = \one(Y
 > \sigma \Phi^{-1}(1-\alpha))$ with certificate $R(Y) = -(Y - \sigma \Phi^{-1}
 (1-\alpha))_+$---i.e.,
 treating the population when the uniformly most accurate confidence lower bound exceeds zero, or equivalently when a \emph{one-sided} test against $\theta \le 0$ rejects.

 On the other hand, if we instead consider \emph{regret} our objective, $L(a,
 \theta) =
 -a\theta - \min(0, -\theta)$, then this loss function is no longer monotone in $\theta$ and
 \cref{thm:inc} no longer applies. Interestingly, as-if optimization using the standard
 two-sided confidence set $[Y - \Phi^{-1}(1-\alpha/2) \sigma, Y + \Phi^{-1}(1-\alpha/2)
 \sigma]$ yields the empirical success rule $\delta(Y) = \one(Y > 0)$, regardless of the
 level $\alpha$. Its corresponding regret certificate is then $R(Y) = \max(0, \Phi^{-1}
 (1-\alpha/2) \sigma - |Y|)$---which is zero whenever $\theta = 0$ is rejected at
 level $\alpha$.
\end{exsq}

\subsection{Benefits to a downstream ambiguity-averse decision-maker}

The guarantee
\eqref{eq:p-cert} is arguably intuitive and may be imposed by fiat as a
 prerequisite for decisions---much like how confidence sets have become standard
 implementations of statistical inference. Indeed, recent work by \citet[Definition 1]
 {galeotti2024robust} imposes exactly such a notion. Yet P-certification is not immediately
 connected to standard decision-theoretic targets such as risk or expected utility.
 Towards establishing such a connection, this subsection analyzes how P-certified
 decisions may provide value to a downstream, ambiguity-averse decision-maker who
 delegates data analysis to an analyst while hoping to minimize expected loss (risk). The
 analyst provides the decision-maker with a P-certified decision, and
 the decision-maker must choose whether to accept the recommendation or to instead
 default to some safe action.

For this subsection, we assume that the loss function is bounded, and without further loss
of generality assume $L(a, \theta) \in [0,1]$ (correspondingly, we may assume certificates are also bounded $R(Y) \in [0,1]$). The decision-maker has access to a
default action $a_0\not\in\mathcal{A}$, which yields a loss $C \in [0,1]$ which is known to the
decision-maker and does not depend on $\theta.$ Our decision-maker decides whether to
adopt the recommended action $\delta(Y)$---incurring an uncertain loss of $L(\delta
(Y), \theta)$---or to implement the default action $a_0$ and incur loss $C$. We model the decision-maker as unfamiliar with the details of the relevant statistical decision problem. Correspondingly, the decision-maker would like to delegate data analysis to an analyst, while requiring worst-case assurances for the performance of the resulting decision.

Let $Q\in\{0,1\}$ be an indicator for whether the decision-maker adopts the analyst's recommendation.  A simple adoption rule is $Q:= \one(R(Y) \le
C)$---adopting whenever the reported certificate is lower than the default cost. The
corresponding loss is $Q L (\delta (Y), \theta) + (1-Q) C$. For such adoption rules,
P-certified decisions provide a worst-case risk bound: the decision-maker incurs  no more
than $u\alpha(1-C)$ excess risk if they never adopt when $R > C$ and adopt with
probability at most $u$ when $R \le C$.

\begin{restatable}{prop}{proppupperbound}
\label{prop:p_upper_bound}
    For $u \in [0,1]$ and any $(\theta, P) \in \M$, any adoption decision $Q$ with $\E[Q
    \mid Y] \le u\one (R
    (Y) \le C)$ has bounded risk
    \[
        \E_P[Q L (\delta (Y), \theta) + (1-Q) C] \le C + u \alpha
        (1-C).
    \]
\end{restatable}

\Cref{prop:p_upper_bound} shows one way in which P-certified decisions provide value to
 ambiguity-averse decision makers: they ensure that the decision-maker does not have
 to tolerate much excess risk so long as they adopt $\delta
 (Y)$ only when $R(Y)$ suggests that this is preferable to the default action.\footnote{Let $\delta_Q(Y)$ denote the recommendation when $Q=1$
and the default action otherwise.
Since $P\{L(\delta_Q
 (Y),\theta)>QR(Y)+(1-Q)C\}\le\alpha$ and $L(\delta_Q(Y), \theta)\le 1$, we  also obtain the
 following $P$-specific bound, which is markedly better if $R (Y)$ can be much smaller
 than $C$:
\[
        \E_P\bk{
            L(\delta_Q(Y), \theta)
        } \le \alpha + \E_P[QR(Y) + (1-Q)C].
    \]
    If $Q = \one(R(Y) \le C)$, this bound simplifies to $\alpha + \E_P
    [\min(R(Y), C)]$.} To further control the excess risk, the decision-maker may either adopt
less, in the sense of choosing a smaller $u$, or ask that the analyst provide
recommendations with higher confidence, requesting a larger $1-\alpha$. Since
\cref{prop:p_upper_bound} applies to all P-certified decisions, including those derived
from trivial confidence sets, naturally the worst-case bound exceeds $C$.

Of course, a decision-maker solely concerned with worst-case risk and facing
sufficient ambiguity would always default to the safe action. We next consider optimal
adoption for a particular form of ambiguity-averse preferences.

Consider a decision-maker who is ambiguity averse and uncertain of the precise loss function, but believes the P-certification property.   The decision maker observes the recommended action and the certificate, $(\delta
(Y), R (Y))$, drawn from a distribution $G$.  The true value of $G$ is the law of $(\delta,
R)$ induced by the true distribution $P$ and the analyst's choices, which in turn depend on the loss $L$.  Let $\ell: \mathcal A \to [0,1]$ denote a generic loss function, which may lie outside the class $\left\{L(\cdot,\theta):\theta\in\Theta\right\}$ induced by the true $L$.  Our next result considers
a decision-maker who entertains \emph{all} pairs $(\ell(\cdot), G)$ satisfying
P-certification:  \[
    \MDMP := \br{(\ell(\cdot), G) :
    G(\ell(\delta) \le R) \ge 1-\alpha, G(\mathcal A \times [0,1]) = 1}.
\]
This model encodes that P-certification is essentially the only thing that the
decision-maker knows.\footnote{To prove that
$(\delta(\cdot), R(\cdot))$ is a P-certified decision, when the loss function $L$ is
money-metric the analyst can in principle offer insurance policies that pay out $1$ unit
of utility if it turns out that $L(\delta(Y), \theta) > R(Y)$, in exchange for $\alpha$
units of insurance premium. The decision-maker can purchase any amount of such policies
before seeing the data. If $ (\delta (\cdot), R (\cdot))$ is indeed a P-certificate,
offering such insurance has nonnegative expected value for the analyst.  As
\citet{neyman1937outline} writes, ``The theoretical statistician constructing [a
confidence interval] may be compared with the organizer of a game of chance in which the
gambler has a certain range of possibilities to choose from while, whatever he actually
chooses, the probability of his winning and thus the probability of the bank losing has
permanently the same value, [$\alpha$].''}

Relative to this model, the decision-maker's \emph{adoption problem} has the adoption loss
\[
    L_{\text{adopt}}(q ; \ell(\cdot), \delta(Y), C) = q\ell(\delta(Y)) + (1-q) C \quad q \in
    [0,1].\numberthis \label{eq:adoption}
\]
A decision rule for this problem is an adoption probability function $q(a,r) \in [0,1]$,
prescribing the probability of adoption as a function of recommended action $a=\delta(Y)$ and
certificate $r=R(Y)$. The corresponding risk function is then \[
    R_{\text{adopt}}(q(\cdot, \cdot); \ell(\cdot), C, G) := \E_G[q(\delta(Y), R(Y)) \ell(\delta(Y)) +
    (1-q(\delta(Y), R(Y))) C].
\]
By \cref{prop:p_upper_bound}, if $q \le u\one(R
(Y) \le C)$, then $R_{\text{adopt}}(q; \ell, C) \le C + u \alpha(1-C)$ for all $(\ell(\cdot), G) \in
\MDMP$.\footnote{Note that this model rationalizes the simple threshold decision $\one(R(Y)
\le C)$
as itself as-if optimization under the confidence set $\br{l : l \le R(Y)}$, valid under all $
(\ell, G)$, for $\ell (\delta (Y))$. }

We suppose that the decision-maker balances between subjective Bayes risk minimization and
worst-case risk over $\MDMP$, following \citet{banerjee2020theory}. The decision-maker has
prior $\pi$ over $(\ell, G)$ and minimizes a weighted average of Bayes risk under $\pi$
and worst-case risk over $\MDMP$: For some weight $\omega \in (0,1)$\[
    \min_{q\in \mathcal Q} \br{(1-\omega) \E_{(\ell, G) \sim \pi}\bk{  R_{\text{adopt}}(q
    (\cdot,
    \cdot);
    \ell (\cdot), C, G) } + \omega \sup_{(\ell, G) \in \MDMP}R_{\text{adopt}}(q(\cdot,
    \cdot);
    \ell (\cdot), C, G)} \numberthis \label{eq:weighted_avg_risk}
\] Equation \eqref{eq:weighted_avg_risk} encodes a trade-off between ambiguity-averse
 preferences \citep{gilboa1989maxmin} and subjective expected utility maximization.

The next proposition characterizes optimal adoption decisions for such preferences. In
particular, among adoption decision rules that solely depend on the reported certificate,
the optimal adoption decision is either to never adopt, or to adopt with a smoothed
version of the threshold rule $\one(R(Y)\le C)$, where the exact smoothing depends on the prior
$\pi$ and the weight $\omega$.

\begin{restatable}{prop}{propoptimaladoption}
\label{prop:optimaladoption}
    Assume:
    \begin{enumerate}
    \item {$C\in(0,1)$ and $\mathcal A$ contains at least two elements.}
        \item
        The class of decision rules $\mathcal Q = \br{q: \mathcal A \times [0,1] \to
    [0,1] :  q (a,r) = q(r)}$
    {consists of rules} that do not depend on the reported action $\delta(Y)$.
    \item The prior $\pi$ in \eqref{eq:weighted_avg_risk} is optimistic in the sense that
    \[
        m(r) := \E[\ell(\delta(Y)) \mid R(Y) = r] \text{ is weakly increasing and }
        m (C)
        < C
    \]
    where the expectation is over the joint distribution of $\ell(\cdot), \delta(Y), R(Y)$
    implied by $(\ell, G) \sim \pi$ and $(\delta(Y), R(Y)) \sim G$. Suppose $m(r_B) = C$
    at some $r_B > C$.
    \end{enumerate}
    {Then there exists a weakly decreasing function
    \(t^*:(0,1)\to\bk{0,\frac{r_B-C}{1-C}}\) with the following property: for every
    \(\omega\in(0,1)\), there is an optimal adoption rule \(q_\omega^*\) for \eqref
     {eq:weighted_avg_risk} that is either identically zero or has the form}
    \[
    {q_\omega^*(r) = \begin{cases}
        1, & \text{ if } r \le C \\
        \min\br{1, \frac{1-C}{r-C} t^*(\omega)} &\text{ if } C < r <
        r_B \\
        0 & \text{ if } r \ge r_B
    \end{cases}.}
    \]
\end{restatable}

\Cref{prop:optimaladoption} shows that the optimal adoption decision is a compromise
 between the simple threshold rule $q(r) = \one(r \le C)$ and the Bayes decision $q_B
 (r) = \one(r \le r_B)$, under beliefs such that subjective expected loss is increasing
 in the reported certificate. When the certificate $R(Y)$ falls below $r_B$, the
 subjective part of the decision-maker's preferences would push for adoption. The presence of
 the ambiguity-averse preferences then reduces adoption when the certificate $R
 (Y)$ is below $r_B$ but above the default cost $C$. In this region, even conditional on
 the certificate being valid, the worst-case loss could still exceed $C$ but expected
 losses are below $C$. The optimal rule balances this possibility---along with the
 possibility that the certificate fails---with expected losses. As the weight $\omega$ on
 the pessimistic part  increases, adoption  in this ``ambiguous'' region $r\in [C, r_B]$
 becomes rarer. Eventually, as $\omega$ increases, the pessimism completely dominates and
 no adoption decision beats the default cost $C$.

Together, \cref{prop:p_upper_bound} and \cref{prop:optimaladoption} show that P-certified
decision rules are approximately compatible with simple adoption decisions for an
ambiguity-averse decision-maker willing to tolerate some increase in her worst-case risk
relative to a safe option. Thus, P-certified decisions offer downstream
decision-makers---including ones less familiar with the statistical environment and  the mapping from parameters to loss---the opportunity to take data-driven decisions without sacrificing (much) safety.

Note that
\cref{prop:p_upper_bound} and \cref{prop:optimaladoption} apply to \emph{any} P-certified
decisions, including the trivial certificate where $R(Y) = 1$ with
probability $1-\alpha,$ independent of $\delta(Y)$. Thus, these results do not require that the analyst be
particularly skilled in data analysis, nor that they necessarily share the same objective as the decision-maker, just that the analyst be able to credibly provide  P-certificates.

\begin{exsq}[Inference on winners, continued]

We next return to \cref{ex:iow} but now consider a setting where a policymaker chooses whether
to adopt an analyst's recommendations on neighborhood choice. We assume that the
policymaker faces a fixed  cost $\kappa$ of making a recommendation at all, for instance
because such recommendations have not been made in the past and this new program needs to
be explained to the public.  To represent this decision-maker's preferences we can
consider the loss function $L(a,\theta)=1-\theta(a),$ and take $C=1-\theta
(a_0)-\kappa$ where for simplicity we treat $\theta(a_0)$ as known.

Suppose the decision-maker is assisted by an analyst who has access to data and knows the
joint distribution of estimation error, where for $a\in\mathcal{A}$, $\frac{X(a)-\theta
(a)}{\sigma(a)}\sim Z(a)$. In this setting, we again take $Y=(X(\cdot),\sigma
(\cdot))$. The analyst may select a recommended action via empirical welfare maximization
and a corresponding P-certificate via either the unstudentized or studentized
projections ($R(Y)$ or $R^*(Y)$ in \cref{ex:iow_ctd}).

Given either $(\delta(Y),R(Y))$ or $(\delta(Y),R^*(Y)),$ the decision-maker must decide
whether to implement $\delta(Y).$  The simple adoption rules $\one\{R
(Y)\le C\}$ and $\one\{R^*(Y)\le C\}$ are size-$\alpha$ tests for the null $L(\delta
(Y),\theta)\ge C$---following them means that the decision-maker adopts the recommended
action only when the analyst rejects the null that it increases loss relative to the
status quo. \Cref{prop:p_upper_bound} shows that such adoption rules control worst-case
excess loss, and \cref{prop:optimaladoption} shows that, under preferences that weight Bayes and worst-case risks, optimal $R$-dependent adoption rules are an
interpolation between the simple rule and the Bayes optimal decision rule.
\end{exsq}

\section{Illustrations of P-certified decisions}\label{sec: p-certificate applications}

We next illustrate how the certified decisions perspective helps interpret empirical findings and construct recommendations in two applications.

\subsection{Creating Moves to Opportunity}

In an experiment known as Creating Moves to Opportunity (CMTO), \citet{bergman2024cmto}
offered low-income Seattle families assistance to move to neighborhoods designated as
\emph{high-mobility}. These high-mobility tracts are roughly the top third of census
tracts in the area, ranked by noisy estimates of upward mobility from the Opportunity
Atlas \citep{chetty2026opportunity}. \citet{bergman2024cmto}, in selecting these Census
tracts, face the challenge of only observing noisy estimates of tract mobility. Different
statistical inference analyses can thus reach different inferential conclusions about the
extent to which the data support designating the selected tracts as
high-mobility.

In one analysis, \citet{andrews2024inference} compare the \emph{average mobility} of the
selected tracts against the average mobility of tracts where the voucher holders currently
reside. They find evidence that, indeed, the select tracts are higher mobility---in the
sense that confidence sets for this difference-of-averages exclude zero. In
another analysis, however, \citet{mogstad2024inference} instead ask which tracts can be confidently
placed in the true top third. They find that, for the 397 tracts in Seattle and
King County, \emph{every tract} is included in the 95\% confidence set for the true
top-third (and also in the confidence set for the bottom-third) of tracts. Thus, at conventional levels, there is insufficient
evidence to support that any given tract---even the one with the highest mobility
estimate---truly belongs to the top third.

Of course, \citet{andrews2024inference} and \citet{mogstad2024inference} consider
different inferential questions, which are in turn relevant for different decision problems. In particular, we show in this section that each inference result corresponds to a P-certificate for an appropriately formulated problem.
The key distinction between these decision problems is the ambiguity set assumed for relocation patterns---that is, for a given selection of target tracts, where will voucher households move from, and where will they move to?
 If prospective movers
are selected at random from the population of housing voucher holders with children and choose
targets uniformly at random among recommended tracts, then the analysis of
\citet{andrews2024inference} is directly reassuring. Instead, if movers could
systematically move from the best unselected tract to the worst selected tract, then
the CMTO program might not actually improve economic mobility at all, corresponding to the conclusion of \citet{mogstad2024inference}. Consequently, as
\citet{mogstad2024inference} write, ``Assumptions on location choices indeed are needed to
be statistically confident [\ldots] that CMTO
would have a positive impact'' (p.514).

To formalize this connection, consider an analyst who recommends a set $S \subset \br{1,\ldots, n}$ of $m := |S| =
\lfloor n/3 \rfloor$ target tracts. We let $\theta_i$ be the true economic mobility of
tract $i$, collected in a vector $\theta \in \R^n$. Suppose the status-quo locations of prospective movers are summarized by a probability vector $p = (p_1,\ldots, p_n)$, while conditional on $S$ movers choose new locations according to a probability vector $q$ supported
on $S$. The average change in economic mobility for such a population of movers is then
$(q-p)'\theta$. For a given class $W(S) \in \Delta^{n-1}$ of $p$ and $V(S) \in \Delta^{n-1}$
of $q$---to be specified---containing the plausible set of movement patterns, we will consider the worst-case
mobility gain over $W(S)$ and $V(S)$: \[ U(S, \theta; W, V) := \min_{p\in W(S), q \in V(S)} (q-p)' \theta. \numberthis \label{eq:general_utility_cmto}
\] This then generates a large menu of decision problems indexed by $W(\cdot), V(\cdot)$.

We winnow down this menu by choosing some parametrizations for $W(\cdot), V(\cdot)$ that
vary in how ``spread out'' movers' current and future locations are. One intuitive
parametrization for plausible targets $V(\cdot)$ is one where movers spread over at
least a fraction $\rho \in [0,1]$ of the selected tracts:\[ V_\rho(S) := \br{q \in
\Delta^{n-1} \text{ supported on $S$}: q_i \le \frac{1}{1 \vee \lceil \rho m \rceil}}. \]
With $\rho=0$, $V_{\rho=0}(S)$ allows all movers to settle in the same tract; with $\rho=1$,
$V_{\rho=1}(S)$ only contains the uniform distribution on $S$. We also consider two
specifications for allowable distributions $W(\cdot)$ of where movers come from.
The first specification fixes $W(S) = \br{b}$, where $b_i$ is the share of housing-voucher
households with children living in tract $i$ in 2018, the year CMTO began.\footnote{We follow \citet{andrews2024inference} in choosing $b$.} The second specification
mirrors $V_\rho(S)$ and considers adversarial choices of $p$: \[ W_\rho(S) := \br{p \in
\Delta^{n-1} \text{ supported on $S^C$}: p_i \le \frac{1}{1\maxwith \lceil \rho (n-m)
\rceil }}, \] where we assume that in the status quo, movers are spread out across a
fraction $\rho$ of the non-selected tracts. We call these two regimes \emph{voucher origins} and \emph{adversarial origins}, respectively.\footnote{\label{fn:intuition_for_rho}To give a little intuition, in the
actual CMTO experiment, the treated families spread across 46 distinct tracts out of 132.
A back-of-the-envelope upper bound for $\rho$ is thus $46/132 \approx 0.35$.
Separately, for each CZ we can compute an estimate of $\rho$ for how spread out voucher holders' current locations are,
$\lfloor 1/\max_i b_i \rfloor/(n-m)$. Its median across the 50 largest Commuting
Zones is about 0.096 (0.143 for King County). }

Having parametrized $V$ and $W$, we now have a menu of decision problems, indexed by $\rho$ and by whether $p$ obeys the voucher-origin or adversarial-origin regime. Over $W_\rho$ and $V_\rho$, we can also directly compute
the minimum in \eqref{eq:general_utility_cmto}. Let $\underline \theta_\rho(S)$ be the
average of the smallest $\rho$ fraction\footnote{Technically, the average of the
smallest-$(\lceil \rho |S| \rceil \vee 1)$ values. } of $\theta_i$ values for $i \in S$,
and let $\bar\theta_\rho(S)$ be the analogous quantity over the largest values. We thus
have: \begin{align*} U_{\text{voucher}}(S, \theta; \rho) &:= U(S, \theta; \br{b}, V_\rho)
=\underline{\theta}_\rho(S) - b'\theta \\ U_{\text{adver}}(S, \theta; \rho) &:= U(S,
\theta; W_\rho, V_\rho) =\underline{\theta}_\rho(S) - \bar\theta_{\rho}(S^C).
\numberthis\label{eq:decision_problem_cmto}\end{align*} We can interpret
\citet{andrews2024inference} and \citet{mogstad2024inference} as targeting different
decision problems in the two ``corners'' of this menu. \citet{andrews2024inference} construct confidence
intervals for $U_{\text{voucher}}(S, \theta; \rho=1)$. Instead, asking if every selected tract is in the top third---parallel to \citet{mogstad2024inference}---examines whether
$U_{\text{adver}}(S, \theta; \rho=0) \ge 0$. Different values of $\rho$ then interpolate between these two corner cases.

We next construct P-certified recommendations for decision problems on this menu. The data
$Y := (Y_1,\ldots, Y_n)$ consists of Opportunity Atlas estimates of $\theta :=
(\theta_1,\ldots, \theta_n)$, which we assume are independent and Gaussian with known
variances: $ Y_i \sim \Norm(\theta_i, \sigma_i^2)$, as in
\citet{andrews2024inference}, \citet{mogstad2024inference} and \citet{chen2026empirical}. To P-certify
decisions, it suffices to construct a confidence set for $\theta$, which then implies
P-certified decisions via as-if optimization.

\begin{figure}[htb]

	\caption{Certificates for mobility gains across the largest Commuting Zones}
	\includegraphics[width=\textwidth]{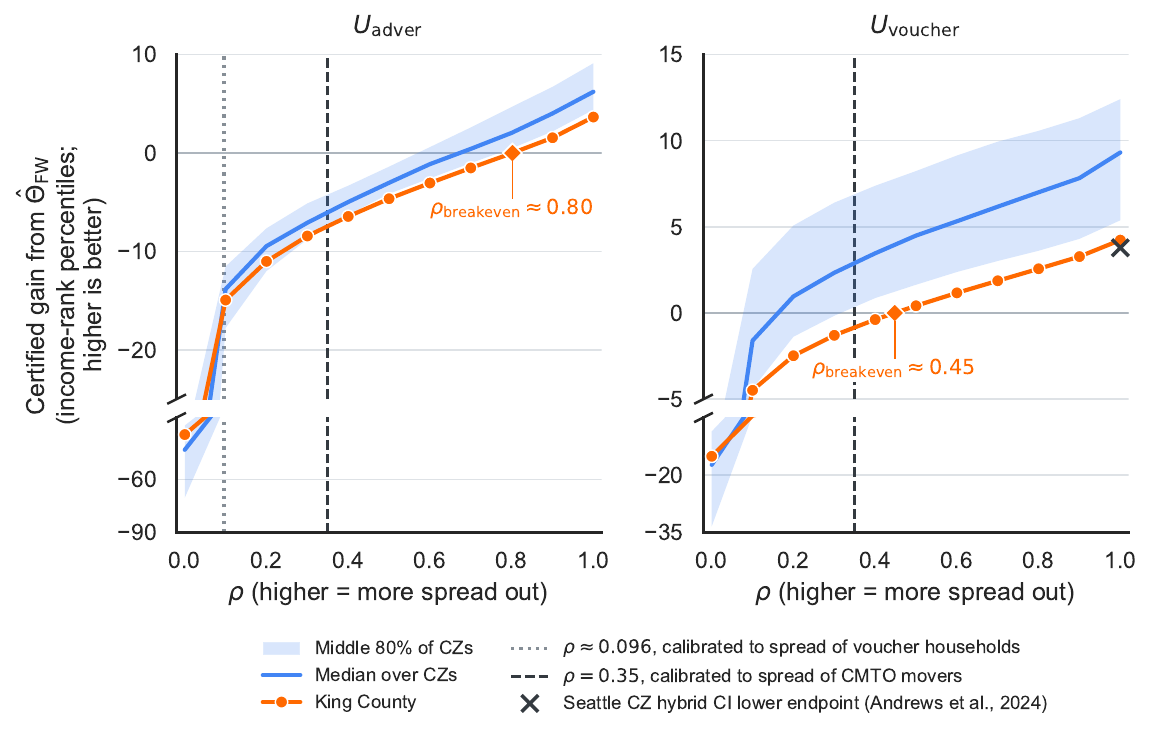}

	\begin{proof}[Notes] We compute certificates $R_{\mathrm{FW}}$ for decision problems
	in \eqref{eq:decision_problem_cmto} for each of the largest 50 CZs,
	following definitions in \citet{andrews2024inference}. For each decision problem, we
	plot the 10th, 50th, and 90th percentiles of these certificates across the 50 CZs. We also
	plot the certificates for King County, WA.
	We also plot the two spread references discussed in \cref{fn:intuition_for_rho}:
	the origin reference $\rho\approx0.096$ in the adversarial-origins panel and
	the destination reference $\rho\approx0.35$ in both panels.
	\end{proof}

	\label{fig:lb_fw_cmto} \end{figure} Given $W, V$, a straightforward \emph{fixed-width}
	confidence set is \[ \hat\Theta_{\text{FW}}(Y)= \left\{ \vartheta\in\mathbb R^n:
	\begin{array}{l} (q-p)'\vartheta\ge(q-p)'Y-c_{\alpha}\\[2pt] \text{for every }S\text{
	with }|S|=m,\; p\in W(S),\;q\in V(S) \end{array} \right\}. \] The critical value
	$c_\alpha$ can be chosen as a value satisfying \[ \P_{\epsilon_i \sim \Norm(0,
	\sigma_i^2)} \underbrace{\left\{ \max_{\substack{S\subseteq\{1,\ldots,n\}\\ |S|=m}} \
	\max_{\substack{p\in W(S)\\q\in V(S)}} (q-p)'\epsilon \le c_\alpha
	\right\}}_{\text{efficiently computable when $(W, V) = (W_\rho, V_\rho)$}}\ge
	1-\alpha. \] This confidence set corresponds to the following P-certified decisions
	for each of the decision problems in \eqref{eq:decision_problem_cmto} (where
	$c_\alpha$ depends on the variant of the decision problem). Note that, since we consider
	utility rather than loss here, certificates are stochastic \emph{lower bounds} of
	utility: \[ \delta_{\text{FW}}(Y) =
	\br{\text{Largest $m$ entries of $Y$}}, \quad R_{\text{FW}}(Y) =
	\underbrace{U(\delta_{\text{FW}}(Y), Y; \rho)}_{\text{plug-in estimate of mobility
	increase}} - c_{\alpha}. \numberthis \label{eq:cmto_certs} \]  Regardless of the decision problem and of the variances of $Y$, we always
	recommend the tracts with the point estimates in the top third. The certificate
	\eqref{eq:cmto_certs} then adjusts a plug-in estimate of mobility increase by a
	penalty for the uncertainty. As such, $\hat\Theta_{\text{FW}}$ has the drawback of
	ignoring precision of the estimates when choosing $S$. In \cref{sub:vw}, we
	consider a variable-width confidence set which penalizes imprecision when as-if
	optimizing the selected tracts; it gives different recommended tracts but similar certificates to those from this fixed-width confidence set.

We plot the certificates \eqref{eq:cmto_certs} in \cref{fig:lb_fw_cmto} across the
menu of decision problems in \eqref{eq:decision_problem_cmto}. We do so both for King
County, WA---as in the actual CMTO experiment---and for the 50 largest Commuting Zones
(CZs). Unsurprisingly, we can certify higher mover mobility gains when movers are more
spread out (higher $\rho$) or under the voucher-origins regime, as opposed to the
adversarial-origins regime. Either direction takes one closer to
\citet{andrews2024inference}'s corner of the menu than to \citet{mogstad2024inference}'s.
Our exercise goes beyond this qualitative description and exactly quantifies how much
robustness to location choices we can demand while still certifying beneficial selections.

At the 5\% level, we can certify a beneficial selection for King County when \(\rho\ge0.45\), even under adversarial destination choices, provided movers' origins follow the observed distribution of voucher holders. In particular, when $\rho =
1$, the fixed-width lower bound is 4.0 rank percentiles, similar to the (two-sided) lower
bound of 3.8 for the hybrid inference method in \citet[Figure V]{andrews2024inference}. At
$\rho \approx 0.35$---calibrated for how spread out the location choices are for those in
the actual CMTO experiment---the certificate is slightly negative. Thus, if movers were to
move adversarially in the CMTO---maintaining the spread of locations---we cannot certify
that the recommendations are net-beneficial at the 5\% level.

If we also require robustness to the movers' current locations, we can only certify a
beneficial selection if $\rho \ge 0.8$. This adversarial-origins regime turns out to be less favorable both
because the plug-in estimate is lower, $U_{\text{adver}}(\delta, Y; \rho) \le
U_{\text{voucher}}(\delta, Y; \rho)$, and because the adversarial regime generally has a
larger penalty $c_\alpha$ for uncertainty. Finally, King County turns out to be a comparatively difficult case: We can
certify beneficial selections for a majority of the CZs at lower values of $\rho$ for
either of the two regimes.  Of course, if movement patterns were not maximally adversarial given the permitted degree of concentration $\rho$, one would obtain more optimistic conclusions.

\subsection{Routing large language models (LLMs)}

Next, we illustrate certified decisions in the context of LLM routing. Arena.ai publishes a
leaderboard of AI models based on user preferences.\footnote
{https://huggingface.co/datasets/lmarena-ai/leaderboard-dataset} Arena's leaderboards are
influential in the AI industry as measures of model quality. We study how a
decision-maker should route queries across LLMs when model quality is unknown. This
setting is similar to that of the empirical exercise in \citet{zrnic2025flexible}, but
more explicitly cast as a decision problem. A decision problem of this sort---at least
after conditioning on sufficiently rich prompt and context characteristics---is faced by
automatic routing services such as Cursor Router, GitHub Copilot Auto, and OpenRouter
Auto.\footnote{See, e.g.,
\url{https://cursor.com/blog/router},
\url{https://docs.github.com/en/copilot/concepts/models/auto-model-selection}, and
\url{https://openrouter.ai/docs/guides/routing/routers/auto-router}. }

Arena's measures are derived from pairwise user comparisons of LLM outputs
\citep{chiang2024chatbot}. The resulting \emph{Arena Scores} are scaled Bradley–Terry estimates
displayed on an Elo-like scale; we refer to them as Elo below (see
\citet{hamilton2023many} for background on the Bradley--Terry model).\footnote
 {Namely, the
 fitted probability of model A beating model B is $1/(1+10^{(Y_B-Y_A)/400})$, per  \url{https://github.com/lmarena/arena-rank}.}
 In this dataset, each LLM $i$ has an estimated Elo score $Y_i$ with standard error
 $\sigma_i$. We will treat $Y_i \sim \Norm (\theta_i,
\sigma_i^2)$ for true Elo score $\theta_i$. In principle, the $Y_i$ are correlated across
$i$ since the estimation noise in $Y$ arises from randomly sampled pairwise
comparisons.
Without access to the variance-covariance matrix of the estimated Elo scores, we adopt the
working model that $Y_i$ are independent across $i$.\footnote{This working model is
accurate if comparisons between models $i$ and $j$ are small proportions of all
comparisons for model $i$ and for model $j$. Without this assumption, one could consider
constructing confidence sets or, later, e-variables robust to dependence (e.g. Bonferroni instead
of max-t confidence sets; averaging valid marginal e-variables). One could also consider a
limited sensitivity analysis, e.g., by considering equicorrelated covariance matrices. }
Arena publishes Elo scores for many tasks; throughout, we use the ``Overall'' category
in ``Text Arena.'' As of August 24, 2026, the top-five (out of all 394 models and, separately, out of the 218 non-proprietary models)
ranked models in this category are shown in \cref{fig:raw-p}.

\begin{figure}[h]
  \centering
  \includegraphics[width=\textwidth]{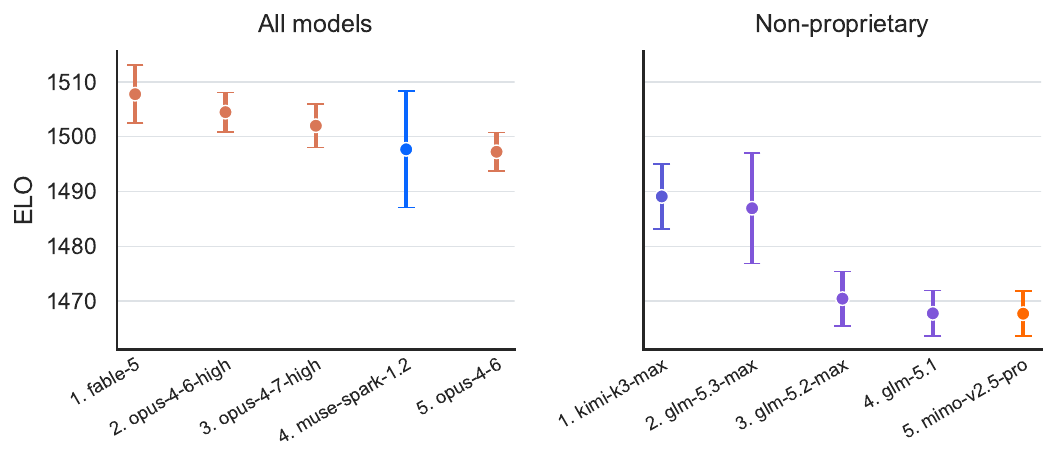}

  \caption{The five highest-ranked models among all models (left) and among
  non-proprietary models (right), as of August 24, 2026, in Text Arena’s Overall category.
  Bars show marginal 95\% confidence intervals.}
  \label{fig:raw-p}
\end{figure}

\begin{figure}[htb]
  \centering
  \includegraphics[width=\textwidth]{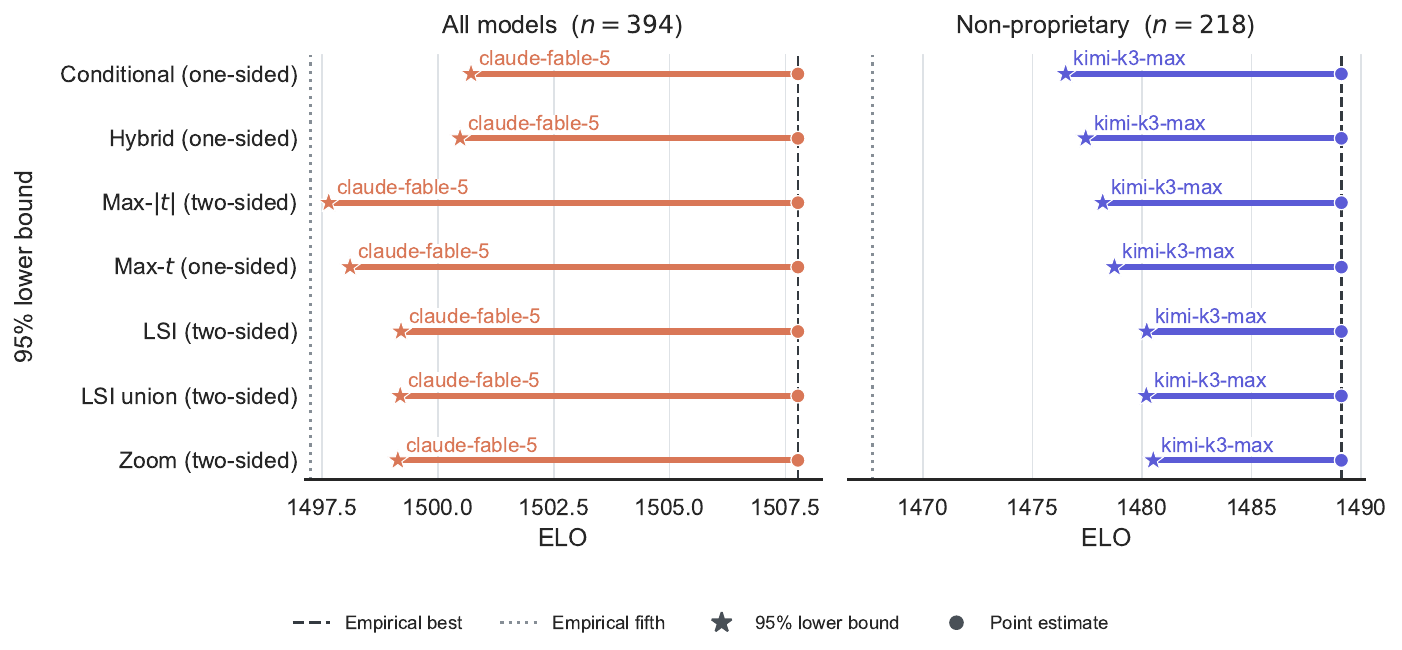}
  \caption{P-certificates for Elo}
  \label{fig:raw-p-cert}
\end{figure}

We first consider a decision maker choosing which LLMs to route queries to with utility
function $u (a, \theta) = \theta_a$ (as with the CMTO exercise, certificates are now stochastic lower bounds for utility). We compare a suite of methods for generating P-certified decisions:
\begin{enumerate}[label=(\roman*)]
    \item The \emph{conditional} and \emph{hybrid} methods of \citet{andrews2024inference}
    take $\delta(Y) = \argmax_a Y_a$ and generate confidence lower bounds for $\theta_
    {\delta(Y)}$.
    \item \emph{Max-t confidence sets} for $\theta = (\theta_1,\ldots,\theta_n)$ imply
    as-if decisions $\delta(Y)$ and corresponding certificates.
    \item \emph{Locally simultaneous inference} (LSI), proposed by
    \citet{zrnic2024locally}, can generate a confidence set for $\theta_{\delta(Y)}$ given
    a user-supplied $\delta(\cdot)$. We apply it to certifying the empirical winner
    $\delta (Y) = \argmax_a Y_a$. It does so by identifying a data-dependent set of
    plausible winners and conducting simultaneous inference over that set. We consider two
    variants of LSI, where one uses a max-t statistic and the other applies a
    Bonferroni/union-bound-type correction.

    \item Finally, the \emph{zoom test}, proposed by \citet{zrnic2025flexible}, similarly
    certifies the empirical winner $\delta(Y) = \argmax_a Y_a$. Although related to LSI,
    it is designed specifically for inference on winners \citep{zrnic2025flexible}.
\end{enumerate} All of these methods, viewed as P-certified decisions, are already as-if
 decisions. Thus they cannot be improved by \cref{thm:dominate}. Since they are all valid
 P-certified decisions, comparing their corresponding P-certificates measures an analogue
 of power tailored to the decision problem (though, of course, on a single data draw).

We compare the P-certificates and decisions of these methods in \cref{fig:raw-p-cert}.
Both among all models and among all non-proprietary models, every method selects the
empirical winner (though this is by construction for methods (i), (iii), and (iv)): Claude Fable
and
Kimi K3, respectively. Every 95\% certificate is above the point estimate of the
fifth-ranked model. For choosing among all models, conditional and hybrid approaches in
\citet{andrews2024inference} appear to produce the tightest certificate. This is possibly
due to the fact that Claude Fable is a clear winner---thus conditioning on the identity of
the empirical winner does not sacrifice too much information. By contrast, among
non-proprietary models, Kimi K3 and GLM 5.3 have similar point estimates, which may make
the conditional procedures' certificates less favorable. In either case, various
certificates differ by at most 4 Elo points, roughly the gap between the point estimates
of Claude Fable and the runner-up Claude Opus 4.6 High.

It is also possible to consider mixed strategies, where an action is a probability vector
and its payoff is $a'\theta$. This corresponds to a decision-maker who will make repeated
LLM-routing decisions and values long-run average quality. Mixing among models can in
principle considerably reduce uncertainty, and thus may be attractive for certified
decisions \citep{chernozhukov2025policylearningconfidence}. One could consider a joint confidence set for
$\br{a'\theta: a \in \Delta([n])}$ \citep{scheffe1953method,kudo1963multivariate} and
choose $\delta(Y) \in
\Delta([n])$ via as-if optimization. However, because there are a large number of models
here, doing so yields a certificate that is about 10 Elo points \emph{lower} for the
setting in the left pane of \cref{fig:raw-p-cert}.\footnote{The Scheff\'e-type critical
value scales as $\sqrt{n}$ when there are $n$ models, whereas max-t critical values scale
like $\sqrt{\log n}$. } Restricting mixtures to the top-ranked models (selected on a holdout sample by data-fission) raises the resulting
certificate, but not appreciably above the bounds shown in \cref{fig:raw-p-cert}.

\begin{figure}[htb]
    \centering

(a) Top five models for price-equivalent Elo under $100\hat\beta_E = \$9.24$
    \includegraphics[width=\textwidth]{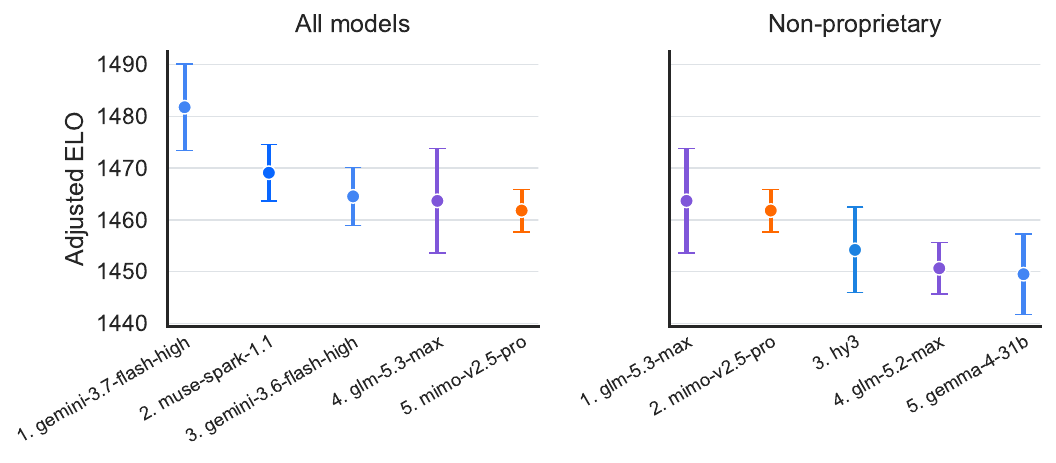}

    (b) Certificates
 \includegraphics[width=\textwidth]{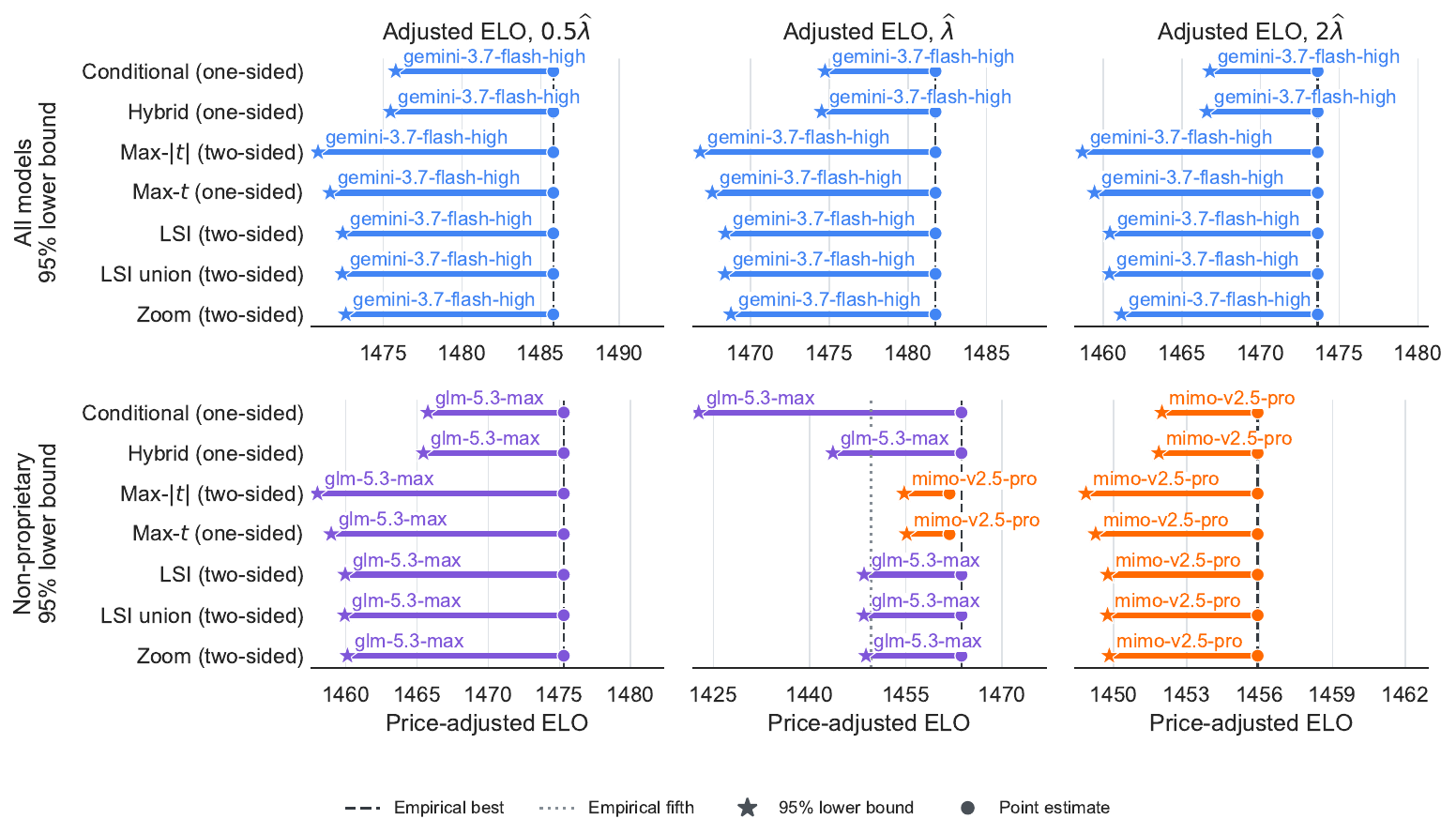}
\caption{P-certificates for price-equivalent Elo}
\label{fig:p_cert_adjusted}
\end{figure}
So far, this exercise ignores LLM-inference costs.
An alternative objective is to maximize \emph{price-equivalent Elo}
$u_\lambda(a,\theta) = \theta_a -
\lambda p_a$, for $p_i$ the price per million tokens for model $i$. We calibrate $\lambda$
 by estimating willingness-to-pay on OpenRouter data. Specifically, we take market share
 data published by OpenRouter for July 24--August 24, 2026,\footnote
 {https://openrouter.ai/rankings, accessed August 24, 2026. This calibration dataset uses
 176 paid standard OpenRouter product variants matched to Arena, each with at least 10,000
 requests, and covers 83\% of monthly chat tokens.} defining a model’s market share as
 the proportion of API requests routed to it. Under a logit random utility model in prices
 and Elo \[ u_j = \underbrace{\beta_E Y_j  + \sigma\epsilon_j + \xi_j}_{\text{WTP for model $j$}} - p_j, \quad \epsilon_j \iid \mathrm{Gumbel}, \quad \E[\xi_j Y_j] =
 \E[\xi_j p_j] = 0,
 \] we can estimate $\beta_E$ by regressing log market shares on an intercept, Elo, and
 prices, and then normalizing the price coefficient to $-1$. Doing so finds that the
 estimated willingness to pay for 100 Elo is $100\hat\beta_E = \$9.24$ per million tokens
 (95\% CI: [\$5.63, \$14.53]). The reciprocal of $\hat\beta_E$ is $\hat\lambda$.

Using $\lambda$ calibrated under this simple model (as well as under sensitivity analyses
by doubling or halving the estimated $\lambda$---while ignoring $\lambda$'s estimation uncertainty),\footnote{Arena and OpenRouter have
somewhat different calculations of prices as different weighted averages of input and
output token prices. We estimate $\lambda$ and calculate price-equivalent Elo both using
the OpenRouter prices. As a result, the price-equivalent exercise restricts to 160 models
with matched Arena--OpenRouter data. } we can repeat the exercise above in
\cref{fig:p_cert_adjusted}. Among all models, every method selects Gemini 3.7 Flash at each of the three values of \(\hat\lambda\), where conditional and hybrid
certificates appear tightest---again, seemingly due to the fact that Gemini 3.7 Flash
clearly outperforms the runner-up.

On the other hand, for non-proprietary models at the estimated willingness to pay, the two
leading models, GLM 5.3 and Mimo 2.5, have more similar point estimates: GLM 5.3 has the
higher estimate but a substantially larger standard error. Consequently, the
max-\(t\)-based procedures select Mimo 2.5 because its lower confidence bound is higher.
This flexibility to optimize $\delta(Y)$ produces tighter certificates than the other
methods designed to certify the empirical winner. At \(0.5\widehat\lambda\), every
procedure selects the GLM-5.3, which has a higher point estimate; at \(2\widehat\lambda\),
every procedure selects the less expensive Mimo 2.5, with the conditional and hybrid
procedures producing the highest lower bounds in both cases.

This exercise illustrates how P-certificates allow us to compare inference procedures
through the guarantees they provide for decisions. In most versions of this exercise,
there is little disagreement over the selected LLM, and thus the tightness of the
resulting guarantees reflects the power of inference-on-winner procedures against
\emph{one-sided} alternatives. Empirically, if there is a clear winner, conditional
inference-type procedures provide tight certificates; otherwise---especially when there is
a clear top-two but no clear winner---locally simultaneous inference procedures are more
favorable, and almost always outperform max-$t$ intervals. However, in the case where the
empirical winner is very noisy (\cref{fig:p_cert_adjusted}(b), middle panel), as-if
optimization against max-$t$ intervals \citep{chernozhukov2025policylearningconfidence} no
longer chooses the empirical winner, and potentially leads to much tighter certificates.
This latter observation illustrates the advantages of viewing decision and inferential goals as linked and optimizing them jointly.

\section{E-certified decisions}
\label{sec:extension_e_certificates}

We have shown that P-certified decisions offer a natural route to combine
(frequentist) statistical inference with decision-making.  However, high-probability
guarantees, as appearing in confidence intervals and hypothesis tests, are not the only
type of frequentist guarantee available. A recent literature in statistics proposes \emph
{e-values} as an alternative and a complement to traditional statistical inference \citep
[see][for a review]{ramdas2024hypothesis}. An \emph{e-variable} against the parameter
value $\theta_0$ is a nonnegative random variable $E(Y,
\theta_0)$
such that \[
    \E_P[E(Y, \theta_0)] \le 1 \text{ for all $P$ such that $(\theta_0, P) \in \mathcal
    M$}.
\]
One can interpret $E(\cdot, \theta_0)$ as the realized dollar payoff of a bet,
priced at \$1, against the hypothesis $H_0: \theta = \theta_0$. Since the expected payoff
under $H_0$ is less than \$1, this bet is (at most) fair. Correspondingly, a large realization of $E
(Y, \theta_0)$ can be interpreted as evidence against $H_0$---in the sense that it is an
unlikely event under $H_0$---much like a small p-value is evidence against $H_0$.

As with confidence sets, we can use e-variables to certify the quality of decisions.
Throughout this section, we assume $L(a,\theta) \ge 0$, but no longer assume $L(a, \theta) \le 1.$

\begin{defn}
\label{defn:e-cert}
     For fixed $\gamma > 0$, a pair $(\delta(\cdot), R(\cdot))$ with $R$ non-negative is an \emph{E-certified decision}
 at multiple $\gamma$    if \[
    \E_P\bk{\frac{L(\delta(Y), \theta)}{R(Y)}} \le \gamma \text{ for all $(\theta, P)\in
    \M$}.
\]
We abbreviate E-certified decisions at multiple 1 as E-certified decisions.
 \end{defn}
 As \cref{defn:e-cert} highlights, $R(Y)$ is a stochastic upper bound for loss in the
 sense that the ratio $L(\delta(Y), \theta)/R(Y)$ has ex ante expectation at most a prespecified $\gamma$.\footnote{We use the conventions that $0/0=0$ and $x/0=\infty$ for $x>0$.} These decisions generalize worst-case expected loss
 considerations: Observe that ordinary worst-case risk for a given decision rule $R =
 \sup_{(\theta, P) \in \mathcal M} \E_P[L(\delta(Y), \theta)] $ is automatically a
 constant-in-$Y$ E-certificate for $\delta(Y)$. From this perspective, minimax risk
 decisions are attractive because they minimize such data-independent E-certificates.

Similar to our previous results, one can derive E-certified decisions using a version of as-if optimization, now
defined using e-variables.
 \citet{grunwald2023posterior} proposes what he terms the E-posterior: if there exists an
  e-variable $E(\cdot, \theta)$ for every $\theta \in \Theta$, the E-posterior is defined
  as $1/E(Y, \theta)$, with the convention $1/0 = \infty$. \citet
  {grunwald2023posterior} proposes computing the action that minimizes the worst-case
  E-posterior-weighted loss, along with the minimized value. Formally, because the infimum over \(\mathcal A\) need not be attained,  let $\mathcal
  C_ {E,\epsilon}$ be the following class of $(\delta, R)$ such that, for some collection
  of e-variables $E(\cdot, \theta)$,
 \begin{align*}
R(Y) \equiv \sup_
 {\theta \in \Theta} \frac{L (\delta(Y),\theta)}{E(Y, \theta)} \le \inf_{a\in \mathcal A}
 \sup_{\theta \in \Theta} \frac{L(a, \theta)}{E
     (Y, \theta)} + \epsilon \numberthis   \label{eq:e-posterior}
 \end{align*}
 Intuitively, for choosing $\delta(\cdot)$, states $\theta$ deemed implausible by the
 data---those with large $E (Y,\theta)$---have their losses downweighted, while those with
 small $E (Y,\theta)$ have their loss upweighted. \citet{grunwald2023posterior} shows that
 $(\delta(Y), R(Y))$ is an E-certified decision at multiple 1.\footnote{Here, our setup is
 slightly different since the minimization is up to a constant $\epsilon$. Nevertheless,
 observe that since $R(Y) \ge L(\delta(Y), \theta) / E(Y,\theta)$, \[
     \E_P[L(\delta(Y), \theta) / R(Y)] \le \E_P[E(Y,\theta)] \le 1.
 \]}

We next show that, as with P-certificates in \cref{thm:dominate}, it is without loss to limit attention to as-if optimization, since the set
of E-certified decisions of the form \eqref{eq:e-posterior} is an essentially complete
class with respect to the dominance order.
\begin{restatable}{prop}{thmedominate}
\label{thm:edominate}
    For any E-certified decision $(\tilde\delta, \tilde R)$ and any $\epsilon > 0$, there
    exists some collection of $e$-variables $E(Y, \theta)$ for which the corresponding
    E-certified decision $ (\delta,R) \in \mathcal C_{E, \epsilon}$ chosen via
    \eqref{eq:e-posterior} weakly dominates $(\tilde \delta, \tilde R)$.
\end{restatable}

As in our results for P-certificates, using E-certificates also provides guarantees when
combined with downstream decision-making. A decision-maker who adopts only when \(R(Y)\le C\) incurs expected loss of at most \(2C\). Unlike our results for
P-certificates, such a guarantee does not require that the loss function be bounded above.

 \begin{restatable}{prop}{propeadoption}
 \label{prop:eadoption}
Given an E-certified decision rule $(\delta, R)$, for any adoption decision $Q
\le \one(R(Y) \le C)$ a.s.,  \[
         \E_P\bk{
             L(\delta_Q(Y), \theta)
         } \le \E_P\bk{
             \max\pr{\frac{L(\delta(Y),\theta)}{R(Y)}, 1}C
         } \le 2C \text{ for all $(\theta, P) \in \M$}
     \]
 \end{restatable}

Motivated by \cref{prop:eadoption}, if we only sought certificates with
$\E_P\left[\max\pr{\frac{L (\delta (Y),\theta)} {R (Y)}, 1}\right] \le 1+\gamma$ to start
with, $R(Y)$ can often be improved by slightly modifying the E-posterior. We detail this
improvement in \cref{asec:truncateep}.

\subsection{Illustration with LLM routing (continued)}
\label{sec:e-empirics}

We illustrate E-certificates in the context of the LLM routing exercise. Since Elo does not have a natural upper bound, loss
(negative utility) does not have a natural lower bound.  E-certification, however, requires a
nonnegative loss function. As a result, we consider
\emph{regret} $L(a, \theta) = \max_{a'} u(a',\theta) - u(a,\theta) \ge 0$ as our
certification target.

We next discuss one way to construct e-variables $E(Y ,\theta)$ to
form E-certificates of regret. Fix a $\theta$ and let $k=k_\theta$ be some element of $\argmax \theta_s$ (or $\argmax
\theta_s - \lambda p_s$)---so that $k$ is the best choice under $\theta$. Intuitively,
operating as if $\theta$ were the true parameter would be sensible, unless there is
evidence that $\E[Y_j] > \E[Y_k]$. We thus should prefer e-variables that are large when
$\E[Y_j] > \E[Y_k]$. For a fixed competitor $j \neq k$, a sensible choice is the
likelihood ratio statistic \[
   0\le  e_{j}(Y, \theta) := \exp\pr{ct_{j}(Y,\theta) - c^2/2}
\]
where $ t_{j}(Y, \theta) := \frac{Y_j - Y_k - (\theta_j -
\theta_k)}{\sqrt{\sigma_j^2 +
    \sigma_k^2}}.
$
The variable $e_j$ is an e-variable because $t_j (Y,\theta) \sim_\theta \Norm(0,1)$ and
hence $\E_\theta [e_ {j} (Y,\theta)] \le 1$. The statistic $e_j(Y,\theta)$ is the
likelihood ratio of $\Norm(c,1)$ against $\Norm (0,1)$ evaluated at data
$t_j(Y,\theta)$---thus tuned to detect alternatives where option $j$'s quality is $c$
standard errors above option $k$.

An e-variable against $\theta$ can thus be constructed by treating $e_j(Y,\theta)$ as
building blocks---{aggregating} them over choices of $j$ and choices of $c$. Such
a \emph{merger} of e-variables is a typical recipe in the e-values literature
\citep{vovk2021values}. In this case, an attractive way of combining them is the
following: For a set of e-variables $E_1(Y, \theta),\ldots, E_J(Y,\theta)$, we can
construct the
\emph{normalized maximum likelihood} \citep{grunwald2024growth}: $E^*(Y,\theta) = \max_j
E_j(Y, \theta) /
\E_\theta[\max_j E_j(Y, \theta)]$, which is an e-variable by construction. We can thus
merge the
e-variables
$e_{j}(Y,\theta)$ by applying this transform twice: first over $j \neq k$, and a second
time over choices of $c \in \br{0.5, 1, 2}$. This process produces e-variables $E(Y,
\theta)$, indexed by $\theta$.

\begin{table}[tb]
    \caption{P- and E-certificates for regret}
    \label{tab:pe_regret}
    \centering
\begin{tabular}{llrlrr}
\toprule
Score & Subset & $|\mathcal A|$ & Selected & $R_P(Y)$ & $R_E(Y)$ \\
\midrule
Raw & All & 394 & claude-fable-5 & 20.45 & 7.29 \\
Raw & Non-proprietary & 218 & kimi-k3-max & 27.26 & 14.01 \\
Adjusted $0.5\widehat\lambda$ & All & 160 & gemini-3.7-flash-high & 22.58 & 7.26 \\
Adjusted $0.5\widehat\lambda$ & Non-proprietary & 87 & glm-5.3-max & 13.81 & 3.66 \\
Adjusted $\widehat\lambda$ & All & 160 & gemini-3.7-flash-high & 15.00 & 1.17 \\
Adjusted $\widehat\lambda$ & Non-proprietary & 87 & glm-5.3-max & 22.50 & 13.73 \\
Adjusted $2\widehat\lambda$ & All & 160 & gemini-3.7-flash-high & 7.91 & 0.11 \\
Adjusted $2\widehat\lambda$ & Non-proprietary & 87 & mimo-v2.5-pro & 17.06 & 6.24 \\
\bottomrule
\end{tabular}

\begin{proof}[Notes]
    In every case, P- and E-certified decisions select the empirical winner. The
    corresponding P- and E-certificates are reported as \(R_P(Y)\) and \(R_E(Y)\),
    respectively. The P-certificates are derived from as-if optimization with 95\% max-t
    confidence sets.
 \end{proof}
\end{table}

We use these e-variables to derive certified decisions for regret in
\cref{tab:pe_regret} (see \cref{sub:e-variable_opt} for how they enable efficient computation). We also show P-certified
decisions derived from max-$|t|$ confidence sets. In every setting that we consider, the
certified decisions agree with the empirical winner. The E-certificates are numerically
lower than the 95\% P-certificates, sometimes substantially, although the two types of
certificate provide different statistical guarantees.

\section{Conclusion}

This paper considers combining statistical inference with statistical decisions. In two
leading modes of frequentist inference, statistical inference translates to \emph
{certified decisions}---decisions paired with certificates of their performance.
We show that it is without loss to base certified decisions on statistical inference procedures, and thus that inferential and decision goals are, at the very least, partially
aligned.
We also
show that these certified decisions provide ambiguity-averse downstream decision-makers
with useful risk guarantees and simple adoption decision rules---which may be implemented
and understood without statistical sophistication on the part of the decision-maker.

\bibliographystyle{ecca}
\bibliography{main}

\appendix

\section{Proofs}

\proppupperbound*

\begin{proof}
    {Fix $(\theta,P)\in\M$, write
    \(A=\{L(\delta(Y),\theta)\le R(Y)\}\), and let
    \(\bar q(Y)=\E[Q\mid Y]\). Since
    \(\bar q(Y)\le u\one\{R(Y)\le C\}\),}
    \[
    {\begin{aligned}
    &\E_P\bk{Q L(\delta(Y),\theta)+(1-Q)C}-C\\
    &\quad=\E_P\bk{\bar q(Y)\bigl(L(\delta(Y),\theta)-C\bigr)}\\
    &\quad\le
    \underbrace{\E_P\bk{\one(A)\bar q(Y)\bigl(L(\delta(Y),\theta)-C\bigr)}}_{\le 0}
    +\E_P\bk{\one(A^c)\bar q(Y)(1-C)}\\
    &\quad\le u(1-C)P(A^c)
    \le u\alpha(1-C).
    \end{aligned}}
    \]
    {The second-to-last line uses that the first expectation is
    nonpositive: whenever \(\bar q(Y)>0\), we have \(R(Y)\le C\), and on \(A\) this
    implies \(L(\delta(Y),\theta)\le C\).}
\end{proof}

\propoptimaladoption*
\begin{proof}
{Let \(\mu\) be
the marginal distribution of \(R(Y)\) under the prior $(\ell, G) \sim \pi$ and $\delta(Y), R(Y)
\sim G$. By
iterated expectations, the subjective excess risk of a rule \(q(r)\) is}
\[
{
\E_{(\ell,G)\sim\pi}\bk{R_{\mathrm{adopt}}(q;\ell,C,G)}-C
=\int_{[0,1]}q(r)\bigl(m(r)-C\bigr)\,\mu(dr).
}
\]

{We first calculate worst-case risk. For a fixed $q$, define}
\[
{
u(q):=\sup_{r\in[0,1]}q(r),
\qquad
b(q):=\sup_{r\in[C,1]}q(r)(r-C).
}
\]
We claim that \[
{
\sup_{(\ell,G)\in\MDMP}R_{\mathrm{adopt}}(q;\ell,C,G)
=C+\alpha(1-C)u(q)+(1-\alpha)b(q).
\numberthis \label{eq:worst_case_adoption}
}
\]
{Take \((\ell,G)\in\MDMP\). On the event $\{\ell(\delta)>R\}$, the excess loss
from adoption is at most $u(q) (1-C)$. On the event $\br{\ell(\delta) \le R}$, the excess
loss is at most $b(q)$, since the excess loss is nonpositive when $R < C$ and is bounded
by $b(q)$ when $R \ge C$.
Since \(b(q)\le u(q)(1-C)\) and  $G(\ell(\delta)>R) \le \alpha$,}
\[
{
R_{\mathrm{adopt}}(q;\ell,C,G)-C
\le\alpha u(q)(1-C)+(1-\alpha)b(q).
}
\]

{This bound is sharp. Choose distinct \(a^-,a^+\in\mathcal A\). For
\(r^-\in[0,1]\) and \(r^+\in[C,1]\), set
\(\ell(a^-)=1\), \(\ell(a^+)=r^+\), and let \(G\) put probability \(\alpha\) on
\((a^-,r^-)\) and probability \(1-\alpha\) on \((a^+,r^+)\). This choice of $(\ell, G)$
belongs to
\(\MDMP\), and its excess risk is}
\[
{
\alpha q(r^-)(1-C)+(1-\alpha)q(r^+)(r^+-C).
}
\]
{Taking sequences \(r^-\) and \(r^+\) approaching the two suprema for $u(q)$
and $b(q)$ therefore
gives \eqref{eq:worst_case_adoption}.}

{Let \(J_\omega(q)\) denote the criterion in
\eqref{eq:weighted_avg_risk}. Combining the preceding displays yields}
\[
{
J_\omega(q)-C
=(1-\omega)\int q(r)(m(r)-C)\,\mu(dr)
+\omega\{\alpha(1-C)u(q)+(1-\alpha)b(q)\}.
\numberthis\label{eq:weighted_excess_risk}
}
\]
{If \(u(q)>0\), write \(q=u(q)\tilde q\), so
\(u(\tilde q)=1\). Both terms on the right side of
\eqref{eq:weighted_excess_risk} are homogeneous of degree one in $q$, and hence}
\[
{
J_\omega(q)-C
=u(q)\bigl(J_\omega(\tilde q)-C\bigr).
}
\]
For any $\tilde q$, if $J_\omega(\tilde q) < C$, then it is optimal to choose $u(q) = 1$.
Otherwise it is optimal to choose $u(q) = 0$ so that $q = 0$ identically. Consequently, an optimal rule can be chosen either to be identically zero or
to satisfy \(u(q)=1\). It remains to characterize the latter rules.

{When $u(q) = 1$, set
\(t=b(q)/(1-C)\). Monotonicity of \(m\), together with
\(m(C)<C=m(r_B)\)}, implies
\[
{
m(r)\le C\text{ for }r\le r_B,
\quad
m(r)\ge C\text{ for }r\ge r_B.
}
\]
Set \(t_B=(r_B-C)/(1-C)\), and for \(t\in[0,t_B]\) define
\[
\bar q_t(r):=
\begin{cases}
1, & r\le C,\\
\min\left\{1,\dfrac{t(1-C)}{r-C}\right\}, & C<r<r_B,\\
0, & r\ge r_B.
\end{cases}
\]
{Fix \(q\) with \(u(q)=1\) and let
\(s:=\min\{t,t_B\}\). We first compare  $q$ and $\bar q_s$ pointwise and claim, for every
\(r\in[0,1]\), that
\[
{
\bigl(\bar q_s(r)-q(r)\bigr)\bigl(m(r)-C\bigr)\le0.
\numberthis\label{eq:order_q}
}
\] For
\(r\le C\),
\(\bar q_s(r)=1\ge q(r)\) and \(m(r)-C<0\). For \(C<r<r_B\), the definition of
\(b(q)\) gives}
\[
{
q(r)\le \min\left\{1,\frac{b(q)}{r-C}\right\}
=\min\left\{1,\frac{t(1-C)}{r-C}\right\}.
}
\]
{If \(t\le t_B\), the right side is \(\bar q_s(r)\). If \(t>t_B\), then
\(s=t_B\) and}
\[
{
\frac{s(1-C)}{r-C}=\frac{r_B-C}{r-C}>1,
}
\]
{so again \(\bar q_s(r)=1\ge q(r)\). Finally, for \(r\ge r_B\),
\(\bar q_s(r)=0\le q(r)\) and \(m(r)-C\ge0\). } This proves \eqref{eq:order_q}.
Integrating shows that \(\bar q_s\) has no greater subjective risk than \(q\):
\[
{
\int_{[0,1]}\bar q_s(r)\bigl(m(r)-C\bigr)\,\mu(dr)
\le
\int_{[0,1]}q(r)\bigl(m(r)-C\bigr)\,\mu(dr).
}
\]

{We next compare worst-case risks. For every \(x\in[0,t_B]\), the definition
of \(\bar q_x\) and the assumption \(r_B>C\) imply}
\begin{align*}
u(\bar q_x)&=1,\\
b(\bar q_x)
&=\sup_{r\in[C,1]}\bar q_x(r)(r-C) =\sup_{C\le r<r_B}\min\{r-C,x(1-C)\} =x(1-C).
\end{align*}
{In particular,}
\[
{
u(\bar q_s)=1=u(q),
\qquad
b(\bar q_s)=s(1-C)
=\min\{b(q),r_B-C\}\le b(q).
}
\]
{Hence \eqref{eq:worst_case_adoption} gives}
\[
{
\begin{aligned}
&\sup_{(\ell,G)\in\MDMP}R_{\mathrm{adopt}}(\bar q_s;\ell,C,G)
-\sup_{(\ell,G)\in\MDMP}R_{\mathrm{adopt}}(q;\ell,C,G)\\
&\quad=\alpha(1-C)\bigl(u(\bar q_s)-u(q)\bigr)
+(1-\alpha)\bigl(b(\bar q_s)-b(q)\bigr)
\le0.
\end{aligned}
}
\]
Thus, for every \(q\) with \(u(q)=1\), we have constructed
\(\bar q_s\) of the desired form, for some \(s\in[0,t_B]\), such that \(u(\bar q_s)=1\)
and
\(J_\omega(\bar q_s)\le J_\omega(q)\). It therefore suffices to minimize
\(J_\omega(\bar q_t)\) over \(t\in[0,t_B]\) and show that the optimal $t$ is weakly
decreasing in $\omega$.

Define the subjective gain from adoption by
\[
H(t):=\int_{[0,1]}\bar q_t(r)\bigl(C-m(r)\bigr)\,\mu(dr).
\]
We may check that the function \(H\) is continuous, nondecreasing, and concave.
Consequently,
\[
{
J_\omega(\bar q_t)-C
=\omega(1-C)\{\alpha+(1-\alpha)t\}-(1-\omega)H(t).
\numberthis\label{eq:phi_adoption}
}
\]
{Continuity and compactness ensure that \(J_\omega(\bar q_t)-C\) has a minimizer. Minimizing \(J_\omega(\bar q_t)-C\) is equivalent to
maximizing}
\[
H(t)-\lambda_\omega t,
\quad
\lambda_\omega:=\frac{\omega}{1-\omega}(1-C)(1-\alpha),
\]
over \(t\in[0,t_B]\). If \(0<\omega_1<\omega_2<1\), and \(t_j\) is a
maximizer for \(H(t)-\lambda_{\omega_j} t\), we have
\[
\bigl(\lambda_{\omega_2}-\lambda_{\omega_1}\bigr)(t_2-t_1)\le0.
\]
Because \(\lambda_\omega\) is strictly increasing, \(t_2\le t_1\). Thus \(t^*(\omega)\)
can be selected to be weakly decreasing in \(\omega\). This proves the result.
\end{proof}

\thmdominate*
\begin{proof}
Consider \[
    \hat\Theta(Y) = \br{
        \theta \in \Theta: L(\tilde\delta(Y), \theta) \le \tilde R(Y)
    }.
\]
Since $(\tilde \delta, \tilde R)$ is P-certified, $\hat\Theta$ is a confidence set:\[
    P(\theta \in \hat\Theta(Y)) = P(L(\tilde\delta(Y), \theta) \le \tilde R(Y)) \ge
    1-\alpha
\]
for all $(\theta, P) \in \mathcal M$.

Consider $(\delta, R) \in \mathcal C_{P, \alpha, \epsilon}$ such that \[
    R(Y) \equiv \sup_{\theta \in \hat\Theta(Y)} L(\delta(Y), \theta) \le \pr{\inf_{a\in
    \mathcal A}
    \sup_
        {\theta \in \hat\Theta (Y)} L(a, \theta) + \epsilon} \minwith \sup_{\theta \in
    \hat\Theta(Y)} L(\tilde \delta(Y), \theta).
\]
Such a choice exists:
\begin{enumerate}[wide]
    \item If $\hat\Theta(Y)$ is empty, then we set $R(Y) = -\infty$ and $\delta(Y) = \tilde\delta(Y)$. The above inequality holds vacuously.
    \item If $\hat\Theta(Y)$ is not empty, then $\inf_{a\in
    \mathcal A}
    \sup_
        {\theta \in \hat\Theta (Y)} L(a, \theta) > -\infty$ by assumption. If $\tilde \delta(Y)$ is already an $\epsilon$-minimizer, \[
           \sup_{\theta \in
           \hat\Theta(Y)} L(\tilde \delta(Y), \theta) \le \inf_{a\in
           \mathcal A}
           \sup_
               {\theta \in \hat\Theta (Y)} L(a, \theta) + \epsilon,
       \]
       then we may choose $\delta(Y) = \tilde \delta(Y)$. Otherwise, if $\tilde \delta(Y)$ is not an $\epsilon$-minimizer, then we can choose  some action with worst-case risk bounded by $\inf_{a\in
           \mathcal A}
           \sup_
               {\theta \in \hat\Theta (Y)} L(a, \theta) + \epsilon$.

\end{enumerate}
By construction, \[
    R(Y) \le \sup_{\theta \in \hat\Theta(Y)} L(\tilde \delta(Y), \theta) \le
    \tilde R (Y).
\]
Since $R(Y)$ and $\tilde{R}(Y)$ are ordered almost surely, they are also ordered in the sense of stochastic dominance. This concludes the proof.
\end{proof}

\thminc*
\begin{proof}
    Consider \[
        \tilde\theta(Y) = \min \br{\theta \in \Theta: L(\tilde \delta(Y), \theta) \le
        \tilde R
        (Y)}
    \]
    On the event that this set is empty, set
    $\tilde\theta(Y)=\max\Theta$.  Otherwise,
 the min is attained by right-continuity and monotonicity. Note that $P(\tilde \theta
(Y) \le \theta) \ge P\{L(\tilde \delta(Y), \theta) \le
        \tilde R
        (Y)\} \ge 1-\alpha$, and thus $\tilde\theta(Y)$ is a $1-\alpha$ confidence lower
        bound for $\theta$.

Again let $R(\theta) = \inf_{a\in \mathcal A} L(a,\theta)$, and note that $R(Y) = R(\hat\theta(Y))$.
If $\tilde\theta(Y)$ does not minimize over an empty set, then
$R(\tilde \theta(Y)) \le L(\tilde\delta(Y),\tilde\theta(Y)) \le \tilde{R} (Y)$. If $\tilde \theta(Y)$ does minimize over  an empty set, then
$R(\tilde\theta(Y))=R(\max\Theta)
\le R(\theta)<r$, so
$\{R(\tilde\theta(Y))\ge r\}$ can occur only if $\tilde \theta(Y)$ minimizes over a non-empty set.
Decreasingness of $L(a,\theta)$ in $\theta$ implies that $R(\theta)$ is decreasing in $\theta.$

Since $\hat\theta$ is a uniformly most accurate confidence lower bound, $P\br{\hat\theta(Y)\le \theta -t}\le P\br{\tilde\theta(Y)\le\theta-t}$ for all $t>0.$  It follows that $P\br{\hat\theta(Y)\land\theta\le t}\le P\br{\tilde\theta(Y)\land \theta\le t}$ for all $t\in\Theta,$ and thus that $\hat\theta(Y)\land\theta$ first order stochastically dominates $\tilde\theta(Y)\land \theta.$  Since first order stochastic dominance is preserved by monotone transformations, it follows that
$R(\tilde\theta(Y)\land \theta)=R(\tilde\theta(Y))\lor R(\theta)$ first order stochastically dominates $R(\hat\theta(Y)\land \theta)=R(\hat\theta(Y))\lor R(\theta),$ and thus that for all $r>R(\theta)$,
\[\P\{R(Y)\ge r\}=\P\{R(\hat\theta(Y))\ge r\}\le \P\{R(\tilde\theta(Y))\ge r\}\le \P\{\tilde{R}(Y)\ge r\}.\]
This completes the proof.
\end{proof}

\thmedominate*
\begin{proof}
    Define \[
        E(Y, \theta) = \frac{L(\tilde \delta(Y), \theta)}{\tilde R(Y)}
    \]
    with the convention that $0/0=0$.
Because $(\tilde \delta, \tilde R)$ is E-certified, $\E_P[E(Y, \theta)] \le 1$ for all $
(\theta, P) \in \M$. Observe that for all $\theta$, \[
    \frac{L(\tilde \delta (Y), \theta)}{E(Y,\theta)} \le \tilde R(Y) \implies \sup_{\theta \in
    \Theta} \frac{L(\tilde \delta (Y), \theta)}{E(Y,\theta)} \le \tilde R(Y) \implies \inf_{a\in \mathcal A}  \sup_{\theta \in
    \Theta} \frac{L(a, \theta)}{E(Y,\theta)} \le \tilde R(Y)
\]
The inequality is due to cases where $L = 0$ but $\tilde R > 0$.

Note that we can pick $(\delta, R) \in \mathcal C_{E, \epsilon}$ such that \[
R(Y) := \sup_{\theta \in \Theta} \frac{L(\delta(Y), \theta)}{E(Y, \theta)} \le \min\br{
\inf_{a\in \mathcal A}  \sup_{\theta \in
\Theta} \frac{L(a, \theta)}{E(Y,\theta)} + \epsilon, \tilde R(Y)
}.
\]
This is because, if $\tilde R(Y)$ is binding on the right-hand side, we can pick $\delta(Y) = \tilde \delta(Y)$. If $\inf_{a\in \mathcal A}  \sup_{\theta \in
\Theta} \frac{L(a, \theta)}{E(Y,\theta)} + \epsilon$ is binding, then we can simply choose any approximate optimum $\delta(Y)$.

By definition, $R(Y) \le \tilde R(Y)$. This completes the proof.
\end{proof}

\propeadoption*
\begin{proof}
    Note that, $P$-almost surely, \[
        L(\delta_Q(Y), \theta) = \frac{L(\delta_Q(Y), \theta)}{QR(Y) + (1-Q) C} \pr{QR(Y)
        + (1-Q) C} \le \max\pr{
            \frac{L(\delta(Y), \theta)}{R(Y)} , 1
        } C.
    \]
    The equality holds since $Q R(Y) + (1-Q) C > 0$ when $L > 0$ by \cref{lemma:positive}.
    The inequality follows from $Q \le \one(R(Y) \le C)$ and thus $QR+(1-Q)C \le C$.

    Now, \[
        \E_{P}\max\pr{
            \frac{L(\delta(Y), \theta)}{R(Y)} , 1
        } \le 1 + \E_P[L/R] \le 2.
    \]
    This completes the proof.
\end{proof}

\begin{lemma}
\label{lemma:positive}
    Let $(\delta(\cdot), R(\cdot))$ be an E-certified decision. For all $(\theta, P) \in
    \M$, \[
        P(R(Y) =0 ,L(\delta(Y),\theta)>0) = 0.
    \]
\end{lemma}
\begin{proof}
Note that if not, then $\E[L/R] = \infty$ since $L > 0, R = 0$ with positive probability. This contradicts E-certification.
\end{proof}

\section{Modifying E-posteriors to optimize post-adoption bound}
\label{asec:truncateep}

Let $\mathcal C_{E,1+\gamma}$ denote the following
class of decision rules: For an $e$-variable $E(\theta, Y)$, define $E_\gamma(\theta, Y) =
\gamma E(\theta, Y) + 1$,\footnote{We can also view $E_\gamma = (1+\gamma) \pr{
    \frac{\gamma}{1+\gamma} E(\theta,  Y) + \frac{1}{1+\gamma} } = (1+\gamma) S^{[\gamma/
     (1+\gamma)]}_{\theta}(Y)$, where $S^{[c]}_ {\theta}(Y) = c E(\theta, Y) +(1-c)$ is
     proposed by \citet{grunwald2023posterior} as a modification to e-variables that
     produces upper-bounded E-posteriors. Our subsequent result gives some guidance on
     how different choices of $\gamma$ impact the resulting guarantee. } and consider
\[
    \delta(Y) \in \argmin_{a \in \mathcal A} \sup_{\theta\in\Theta} \frac{L(a, \theta)}
    {E_\gamma
    (\theta, Y)} \quad R(Y) = \min_{a \in \mathcal A} \sup_{\theta\in\Theta} \frac{L(a,
    \theta)}{E_\gamma (\theta, Y)}.\numberthis \label{eq:truncated_e}
\]
In this section, let us assume $\mathcal A$ is finite for simplicity, and thus the minimum
always exists. Since $E_\gamma/(1+\gamma)$ is an e-variable, such a decision rule is
E-certified with multiple $1+\gamma$.

We show that adoption rules with $Q \le \one(R(Y) \le C)$ entail risk at most $(1+\gamma) C$ post adoption, in contrast to
\cref{prop:eadoption}. This guarantee holds because \[
    \E_P
             \left[\max\pr{\frac{L(\delta(Y),\theta)}{R(Y)}, 1}\right] \le 1+\gamma.
\]
Moreover, for all decisions $(\delta, R)$ satisfying the above, decision rules in
$\mathcal C_{E,1+\gamma}$ are essentially complete with respect to weak
dominance.

 \begin{restatable}{prop}{proptruncated}
 \label{prop:truncated}
     Given $(\delta, R) \in \mathcal C_{E,1+\gamma}$, we have that for any adoption
     decision $Q \le \one(R(Y) \le C)$, \[
         \E_P\bk{
             L(\delta_Q(Y), \theta)
         } \le \E_P\bk{
             \max\pr{\frac{L(\delta(Y),\theta)}{R(Y)}, 1}C
         } \le (1+\gamma)C \text{ for all $(\theta, P) \in \M$}
     \]
     Moreover, for any pair $(\tilde \delta, \tilde R)$ such that $
         \E_P\bk{\max\pr{\frac{L(\tilde \delta(Y),\theta)}{\tilde R(Y)}, 1}} \le 1+ \gamma
     $, there exists some $(\delta, R) \in \mathcal C_{E,1+\gamma}$ that weakly dominates it.
 \end{restatable}

\begin{proof}
    Observe that \[
        L(\delta_Q(Y), \theta) = \frac{L(\delta_Q(Y), \theta)}{QR(Y) + (1-Q) C} \pr{QR(Y)
        + (1-Q) C} \le \max\pr{
            \frac{L(\delta(Y), \theta)}{R(Y)} , 1
        } C.
    \]
    Since \[
        \E_P\bk{
            \max\pr{
            \frac{L(\delta(Y), \theta)}{R(Y)} , 1
        }
        } \le \E_P\bk{
            \max\pr{
                \frac{L(\delta(Y), \theta)}{L(\delta(Y), \theta)/E_\gamma(Y,\theta)} , 1
            }
        } \le \E_P\bk{E_\gamma(Y,\theta)} \le 1+\gamma,
    \]
    this proves the first claim.

    For the second claim, let \[
        E_\gamma(Y, \theta) = \max\pr{
            \frac{L(\tilde \delta(Y), \theta)}{\tilde R(Y)} , 1
        }.
    \]
    Since its expectation is bounded by $1+\gamma$, $(E_\gamma(Y,\theta) - 1)/ \gamma$ is an
    e-variable. Thus
    there
    is some $(\delta, R) \in \mathcal C_{E,1+ \gamma}$ that corresponds to $E_\gamma(Y,\theta)$. Note
    that \[
        R(Y) = \min_{a\in \mathcal A} \sup_{\theta\in\Theta} \frac{L(a,\theta)}{E_\gamma
        (Y,\theta)} \le \sup_{\theta\in \Theta} \min\pr{\tilde R(Y)\cdot L(\tilde\delta
        (Y),\theta)/L
        (\tilde\delta(Y), \theta), L(\tilde\delta(Y), \theta)} \le \tilde R(Y).
    \]
    This completes the proof.
\end{proof}

\section{Additional details of empirical applications}

\subsection{Variable-width confidence sets for the CMTO application}
\label{sub:vw}
We could consider a variable-width
version of $\hat \Theta_{\mathrm{FW}}$\footnote{Here, $\ext V$ denotes the extreme points of a set $V$. For $V_\rho(S)$,
the extreme points are those movement patterns supported on $S$ with share exactly
$1/(1\vee \lceil \rho m \rceil)$ at exactly $1\vee \lceil \rho m \rceil$ locations. For a \emph{fixed} $\lambda_\alpha$, removing $\ext$ in $\hat\Theta_{\text{VW}}$ results in the same as-if decisions as $\hat\Theta_{\text{VW}}$. But removing $\ext$ would demand a higher critical value, and thus would result in worse certificates than $\hat\Theta_{\text{VW}}$.} \[
\hat \Theta_{\text{VW}} := \left\{ \vartheta\in\mathbb R^n: \begin{array}{l}
(q-p)'\vartheta\ge(q-p)'Y- \lambda_\alpha \sum_{i} |q_i - p_i| \sigma_i\\[2pt] \text{for
every }S\text{ with }|S|=m,\; p\in \ext W(S),\;q\in \ext V(S) \end{array} \right\} \]
where the critical value $\lambda_\alpha$ is chosen such that \[ \P_{\epsilon_i \sim
\Norm(0, \sigma_i^2)}\left\{ \max_{\substack{S\subseteq\{1,\ldots,n\}\\|S|=m}} \
\max_{\substack{p\in\operatorname{ext}W(S)\\ q\in\operatorname{ext}V(S),\ q\ne p}}
\frac{(q-p)'\epsilon} {\sum_i|q_i-p_i|\sigma_i} \le\lambda_\alpha \right\}\ge1-\alpha. \]
Here, using $\sum_i |q_i - p_i| \sigma_i$ as a proxy for the sampling noise in $(q-p)'Y$
rather than $\sqrt{\sum_i (q_i - p_i)^2 \sigma_i^2}$ is primarily motivated by
computational tractability. Qualitatively, as-if optimization under $\hat\Theta_{\text{VW}}$ chooses tracts by penalizing those with high $\sigma_i$. We compute similar certificates in \cref{fig:lb_vw_cmto}.

\begin{figure}[tb]
	\centering
	\includegraphics[width=\textwidth]{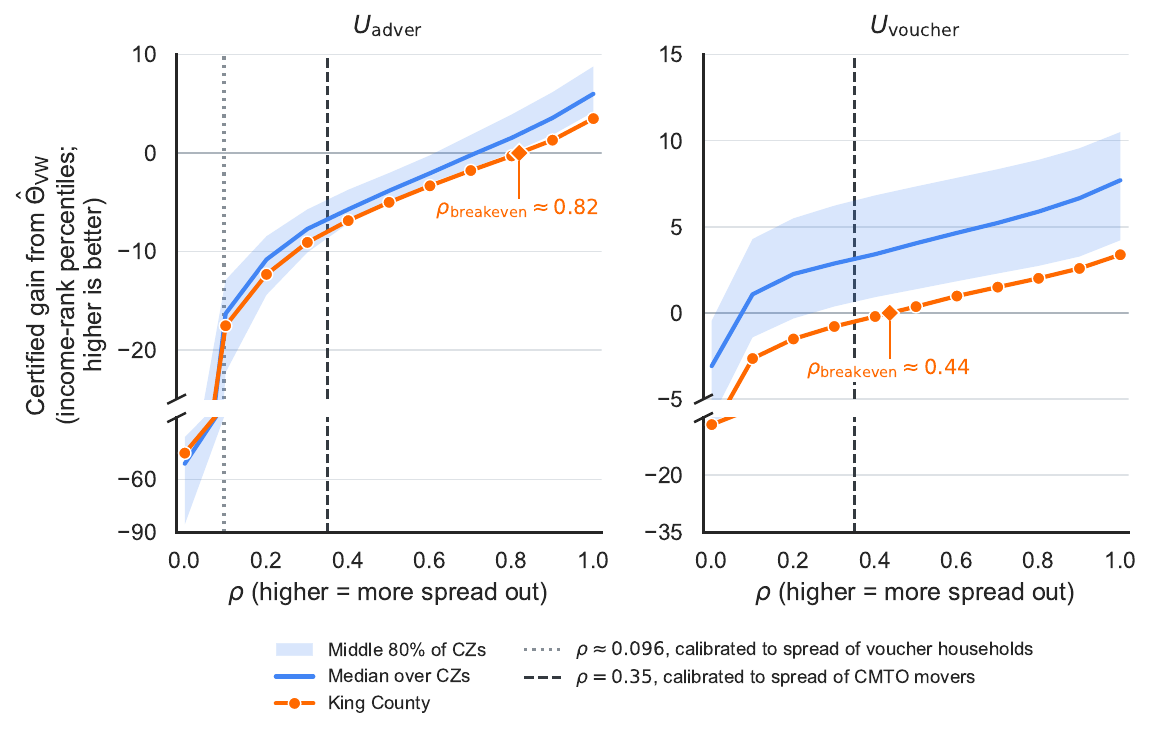}
	\caption{Analogue of \cref{fig:lb_fw_cmto} using $\hat\Theta_{\mathrm{VW}}$}
	\label{fig:lb_vw_cmto}
\end{figure}

While the certificates are similar to those in \cref{fig:lb_fw_cmto}, $\hat\Theta_{\mathrm{VW}}$ and $\hat\Theta_{\mathrm{FW}}$ make different recommendations when as-if optimized against. We can also use the two confidence sets to certify \emph{each other}'s recommendations in \cref{tab:table3}. We see that the resulting certificates are different, indicating that the selected tracts have some discrepancy between the two confidence sets.
\begin{table}[tb]
	\caption{Certificates under $(\hat\Theta_{\mathrm{FW}}, \hat\Theta_{\mathrm{VW}})$ for recommended decisions using either confidence set}
	\label{tab:table3}
	\centering

	\begin{tabular}{@{}lll@{}}
\toprule
Selected set & FW certificate & VW certificate \\
\midrule
\multicolumn{3}{@{}l}{\emph{Median of 50 CZs, adversarial origins, $\rho = 0.1$}} \\
FW-selected & $-13.8$ & $-16.7$ \\
VW-selected & $-14.4$ & $-16.3$ \\
\midrule
\multicolumn{3}{@{}l}{\emph{King County, adversarial origins, $\rho = 0.1$}} \\
FW-selected & $-14.9$ & $-18.4$ \\
VW-selected & $-15.9$ & $-17.5$ \\
\midrule
\multicolumn{3}{@{}l}{\emph{Median of 50 CZs, voucher origins, $\rho = 0.1$}} \\
FW-selected & $-1.6$ & $-1.7$ \\
VW-selected & $-3.0$ & $1.1$ \\
\midrule
\multicolumn{3}{@{}l}{\emph{King County, voucher origins, $\rho = 0.1$}} \\
FW-selected & $-4.5$ & $-5.4$ \\
VW-selected & $-5.9$ & $-2.6$ \\
\midrule
\multicolumn{3}{@{}l}{\emph{Median of 50 CZs, adversarial origins, $\rho = 0.5$}} \\
FW-selected & $-3.0$ & $-3.9$ \\
VW-selected & $-3.1$ & $-3.8$ \\
\midrule
\multicolumn{3}{@{}l}{\emph{King County, adversarial origins, $\rho = 0.5$}} \\
FW-selected & $-4.6$ & $-5.1$ \\
VW-selected & $-4.8$ & $-5.0$ \\
\midrule
\multicolumn{3}{@{}l}{\emph{Median of 50 CZs, voucher origins, $\rho = 0.5$}} \\
FW-selected & $4.5$ & $3.1$ \\
VW-selected & $3.8$ & $4.0$ \\
\midrule
\multicolumn{3}{@{}l}{\emph{King County, voucher origins, $\rho = 0.5$}} \\
FW-selected & $0.4$ & $-0.6$ \\
VW-selected & $-0.5$ & $0.4$ \\
\bottomrule
\end{tabular}

\end{table}

\subsection{Details of as-if optimization with e-variables}
\label{sub:e-variable_opt}
In particular, with the choice of e-variables in \cref{sec:e-empirics}, the as-if
optimization following \citet{grunwald2023posterior} simplifies considerably due to the
observation that \[
    \sup_{\theta} \frac{L(a, \theta)}{E(Y, \theta)} = \max_{k \neq a} \sqrt{\sigma_a^2 +
    \sigma_k^2} \sup_{t \ge \max_{j\neq k} \frac{Y_{j} - Y_k}{\sqrt{\sigma_j^2 +
    \sigma_k^2}}}
\frac{t - \frac{Y_a - Y_k}{\sqrt{\sigma_a^2 + \sigma_k^2}}}{G_k(t)}
\]
for a known non-decreasing function $G_k(\cdot)$.

This follows from two observations. First, if $k$ is the best arm under $\theta$, then $E
(Y,\theta) = G_k(\max_{j\neq k} t_j(Y, \theta))$ for some non-decreasing $G_k
(\cdot)$. Second, for a fixed $k$, let $t
(\theta, k) =
 \max_{j\neq k}
t_j(Y, \theta)$. We have that \[\Theta = \bigcup_{k\in[n]} \bigcup_{t \ge
\max_ {j\neq k} \frac{Y_{j} - Y_k}{\sqrt{\sigma_j^2 +
    \sigma_k^2}}}
\br{\theta : t(\theta, k) = t, k(\theta) = k}.
    \] Maximizing $L(a, \theta)/E(Y,\theta)$ over the inner set of $\theta$, applying the
     first observation, yields the closed form $G_k(t)^{-1} \sqrt
     {\sigma_a^2 + \sigma_k^2} \pr{t -\frac{Y_a - Y_k}{
    \sqrt{\sigma_a^2 +
    \sigma_k^2}} }$.

\end{document}